%% file: 00Main.tex
\documentclass[letterpaper,journal]{IEEEtran}
\usepackage{amsmath, amssymb, amsthm, bm, mathtools}
\usepackage{enumitem}
\usepackage[final]{graphicx}
\usepackage{epstopdf}
\usepackage{cite}

\usepackage{tikz}
\usetikzlibrary{matrix,fit,calc,positioning,arrows.meta}
\usepackage[caption=false,font=footnotesize]{subfig}
\usepackage{xcolor}

\input{preamble/macros.tex}
\input{preamble/letterfonts.tex}

\input{preamble/preamble.tex}

\newtheorem{theorem}{Theorem}
\newtheorem{proposition}{Proposition}
\newtheorem{lemma}{Lemma}

\theoremstyle{remark}

\newtheorem{definition}{Definition}
\newcommand{\R}{\mathbb{R}}
\renewcommand{\EE}{\mathbb{E}}
\newcommand{\Prob}{\mathbb{P}}
\newcommand{\hSigma}{\hat{\boldsymbol{\Sigma}}}
\newcommand{\hGamma}{\hat{\boldsymbol{\Gamma}}}
\newcommand{\bSigma}{\boldsymbol{\Sigma}}

\renewcommand{\supp}{\mathrm{supp}}

\usepackage{booktabs}

\newcommand{\bfxi}{\bm{\xi}}

\newcommand{\anglesin}{\sin\angle}

\title{Sparse Principal Component Analysis with Energy Profile Dependent Sample Complexity}
\author{Mengchu~Xu,~\IEEEmembership{Member,~IEEE,}
        Jian~Wang,~\IEEEmembership{Member,~IEEE,}
        and~Yonina~C.~Eldar,~\IEEEmembership{Fellow,~IEEE}%
\thanks{The work of Mengchu Xu and Yonina C. Eldar was supported in part by the European Research Council (ERC) through the European Union's Horizon 2020 Research and Innovation Program under Grant 101000967 and in part by the Israel Science Foundation under Grant 536/22. The work of Jian Wang was supported in part by the National Key R\&D Program of China under Grant 2024YFF0505601 and in part by the National Natural Science Foundation of China under Grant 62471147. A preliminary version of this work was presented at the 2026 IEEE International Symposium on Information Theory (ISIT), Guangzhou, China, June 28--July 3, 2026. (Corresponding author: Mengchu Xu.)}%
\thanks{Mengchu Xu is with the School of Data Science, Fudan University, Shanghai 200433, China. He was with the Faculty of Mathematics and Computer Science, Weizmann Institute of Science, Rehovot 7610001, Israel, when this work was performed (e-mail: mcxu@fudan.edu.cn; mengchu.xu@weizmann.ac.il).}%
\thanks{Jian Wang is with the School of Data Science, Fudan University, Shanghai 200433, China (e-mail: jian\_wang@fudan.edu.cn).}%
\thanks{Yonina C. Eldar is with the Faculty of Mathematics and Computer Science, Weizmann Institute of Science, Rehovot 7610001, Israel (e-mail: yonina.eldar@weizmann.ac.il), and also with the Department of Electrical and Computer Engineering, Northeastern University, Boston, MA 02115 USA (e-mail: y.eldar@northeastern.edu).}%
}
\date{}

\begin{document}
\maketitle

\begin{abstract}
  We study sparse principal component analysis in the high-dimensional, sample-limited regime, aiming to recover a leading component supported on a few coordinates. Despite extensive progress, most methods and analyses are tailored to the flat-spike case, offering little guidance when spike energy is unevenly distributed across the support. Motivated by this, we propose Spectral Energy Pursuit (SEP), an effective iterative scheme that repeatedly screens and reselects coordinates, with a sample complexity that adapts to the energy profile. We develop our framework around a structure function \(s(p)\) that quantifies how spike energy accumulates over its top \(p\) entries. To our knowledge, SEP is the first polynomial-time SPCA method with a sample-complexity guarantee governed by the full energy profile: it succeeds with \(m\gtrsim \max_{1\le p\le k} p\,s^2(p)\,\log n\) samples, recovering the classical \(k^2\log n\) rate for flat spikes and improving to \(k\log n\) for sufficiently concentrated profiles. As a lightweight post-processing, a single truncated power iteration is proven to enable the final estimator to attain a uniform statistical error bound. Empirical simulations using a flat profile and offset-regularized decaying profiles validate that SEP adapts to profile structure without profile-specific tuning and outperforms existing algorithms.
\end{abstract}


\input{01Introduction.tex}
\input{03Algorithm.tex}
\input{04Theory.tex}
\input{041Proof.tex}
\input{05Discussion.tex}
\input{06Simulations.tex}
\input{07Conclusion.tex}
\input{08Appendix.tex}

\bibliographystyle{IEEEtran}
\bibliography{IEEEabrv2025, spca, pr}
\input{09Biographies.tex}

\end{document}

%% file: preamble/macros.tex
\newcommand{\iid}{\textit{i.i.d.\ }}

%% file: preamble/letterfonts.tex
\newcommand{\PP}[1][]{\ensuremath{\ifstrempty{#1}{\mathbb{P}}{\mathbb{P}^{#1}}}}

\newcommand{\EE}{\ensuremath{\mathbb{E}}}

\newcommand{\bfA}{\mathbf{A}}	
	
\newcommand{\bfE}{\mathbf{E}}	
	
\newcommand{\bfI}{\mathbf{I}}

\newcommand{\bfW}{\mathbf{W}}

\newcommand{\bfa}{\mathbf{a}}	
	
\newcommand{\bfe}{\mathbf{e}}

	\newcommand{\bfr}{\mathbf{r}}
	
\newcommand{\bfu}{\mathbf{u}}	\newcommand{\bfv}{\mathbf{v}}
\newcommand{\bfw}{\mathbf{w}}	\newcommand{\bfx}{\mathbf{x}}
\newcommand{\bfy}{\mathbf{y}}	\newcommand{\bfz}{\mathbf{z}}

\newcommand{\mcE}{\mathcal{E}}	
	\newcommand{\mcH}{\mathcal{H}}

	\newcommand{\mcN}{\mathcal{N}}
\newcommand{\mcO}{\mathcal{O}}	\newcommand{\mcP}{\mathcal{P}}

\newcommand{\supp}[1]{\operatorname{supp}\left(#1\right)}

%% file: preamble/preamble.tex
\usepackage{amsmath,amsfonts,amsthm,amssymb,mathtools}
\usepackage{xfrac}
\usepackage[makeroom]{cancel}
\usepackage{enumitem}
\usepackage{theoremref}
\usepackage{xparse}

\usepackage[most,many,breakable]{tcolorbox}
\usepackage{varwidth}
\usepackage{etoolbox}
\usepackage{nameref}
\usepackage{multicol,array}
\usepackage{algorithm,algpseudocode}
\usepackage{comment}
\usepackage{import}
\usepackage{xifthen}
\usepackage{pdfpages}
\usepackage{transparent}
\usepackage{chngcntr}
\usepackage{tikz}
\usepackage{titletoc}

\usepackage{empheq}
\usepackage{hyperref}
\usepackage{bookmark}
\usepackage{cleveref}

\crefname{figure}{Fig.}{Figs.}
\Crefname{figure}{Fig.}{Figs.}
\crefname{subfigure}{Fig.}{Figs.}
\Crefname{subfigure}{Fig.}{Figs.}
\crefname{table}{TABLE}{TABLES}
\Crefname{table}{TABLE}{TABLES}
\crefname{section}{Section}{Sections}
\Crefname{section}{Section}{Sections}
\crefname{subsection}{Section}{Sections}
\Crefname{subsection}{Section}{Sections}
\crefname{algorithm}{Algorithm}{Algorithms}
\Crefname{algorithm}{Algorithm}{Algorithms}
\crefname{theorem}{Theorem}{Theorems}
\Crefname{theorem}{Theorem}{Theorems}
\crefname{proposition}{Proposition}{Propositions}
\Crefname{proposition}{Proposition}{Propositions}
\crefname{corollary}{Corollary}{Corollaries}
\Crefname{corollary}{Corollary}{Corollaries}
\crefname{appendix}{Appendix}{Appendices}
\Crefname{appendix}{Appendix}{Appendices}

\counterwithin*{ALG@line}{algorithm}

\definecolor{doc}{RGB}{0,60,110}
\definecolor{myg}{RGB}{56, 140, 70}
\definecolor{myb}{RGB}{45, 111, 177}
\definecolor{myr}{RGB}{199, 68, 64}
\definecolor{mytheorembg}{HTML}{F2F2F9}
\definecolor{mytheoremfr}{HTML}{00007B}
\definecolor{mylemmabg}{HTML}{FFFAF8}
\definecolor{mylemmafr}{HTML}{983b0f}
\definecolor{mypropbg}{HTML}{f2fbfc}
\definecolor{mypropfr}{HTML}{191971}
\definecolor{myexamplebg}{HTML}{F2FBF8}
\definecolor{myexamplefr}{HTML}{88D6D1}
\definecolor{myexampleti}{HTML}{2A7F7F}
\definecolor{mydefinitbg}{HTML}{E5E5FF}
\definecolor{mydefinitfr}{HTML}{3F3FA3}
\definecolor{notesgreen}{RGB}{0,162,0}
\definecolor{myp}{RGB}{197, 92, 212}
\definecolor{mygr}{HTML}{2C3338}
\definecolor{myred}{RGB}{127,0,0}
\definecolor{DodgerBlue}{RGB}{30, 144, 255}
\definecolor{myyellow}{RGB}{169,121,69}
\definecolor{myexercisebg}{HTML}{F2FBF8}
\definecolor{myexercisefg}{HTML}{88D6D1}

%% file: 01Introduction.tex
\section{Introduction}\label{sec:introduction}

Principal Component Analysis (PCA)~\cite{EckartYoung1936MatrixApprox,Jolliffe2002PCA} is a cornerstone of multivariate statistics and machine learning and has numerous applications in data analysis and dimensionality reduction~\cite{TurkPentland1991Eigenfaces}. In high dimensions with a limited number of samples, however, classical PCA can be statistically inefficient and unreliable. Sparse PCA (SPCA) addresses this statistical inconsistency by seeking a leading component with small support~\cite{Zou2006SPCA, Guan2009SPCAProbabilistic, Lee2010SPCABiclustering, GuerraUrzola2021SPCAGuide}. In its simplest form, we observe \(m\) samples \(\bfx_1,\dots,\bfx_m\in\mathbb{R}^n\) drawn \iid from a centered Gaussian distribution \(\mcN(\mathbf{0}, \bSigma)\) with covariance structure
\begin{equation}
    \bSigma = \bfI_n + \theta\, \bfv\bfv^\top,\qquad \|\bfv\|_2=1,\quad \|\bfv\|_0\le k,
\end{equation}
where \(\theta>0\) quantifies the spike strength.
The goal of SPCA is to estimate the leading eigenvector \(\bfv\) of \(\bSigma\) under the assumption that the spike \(\bfv\) has at most \(k\) nonzero entries.

A mature theory now characterizes the fundamental limits of SPCA under the single-spike model: the minimax sample complexity for consistent direction estimation scales as \(m \asymp k\log n\) when the spike is \(k\)-sparse~\cite{VuLei2012MinimaxSPCA, WangLuLiu2014MinimaxSPCA, Johnstone2009SPCAConsistency, BerthetRigollet2013aSPCADetection}. Therefore, this bound is often called the statistical lower bound for SPCA, and can be achieved by exhaustive search over all \(\binom{n}{k}\) possible supports~\cite{BerthetRigollet2013aSPCADetection, Liu2021Towards}, i.e., solving the following NP-hard optimization problem:
\begin{equation}\label{eq:spca-problem}
    \hat\bfv = \underset{\bfw: \|\bfw\|_2=1, \|\bfw\|_0\le k}{\arg \max} \bfw^\top \hSigma \bfw,
\end{equation}
where \(\hSigma = \frac{1}{m}\sum_{i=1}^m \bfx_i \bfx_i^\top\) is the sample covariance matrix. However, efficiently achieving the optimal sample complexity via practical polynomial-time algorithms remains challenging. Classical screening/thresholding methods select high-variance or correlated coordinates and then run PCA on the restricted submatrix~\cite{Johnstone2009SPCAConsistency, Ma2013SPCAIT}. Their guarantees typically scale as \(m\gtrsim k^2\log n\). Semidefinite relaxations (SDP)~\cite{dAspremont2007PCASDP,Amini2009SPCASDP} can achieve \(m\gtrsim k\log n\) when the solution is rank one, but ensuring rank one requires \(m\gtrsim k^2\log n\)~\cite{Krauthgamer2015SDPSPCAInfoLimit}, and solving large-scale SDPs is computationally demanding. The gap between the statistical lower bound \(k\log n\) and the practical sample complexity \(k^2\log n\) is believed to be fundamental, as improving upon it would imply breakthroughs in other well-known hard problems, such as Planted Clique~\cite{Feige2000HiddenClique,BerthetRigollet2013bSPCADetectionLB, Wang2016SPCACompStatTradeoff}.

The algorithms and guarantees stated above hold uniformly over all \(k\)-sparse spikes and thus are governed by worst-case performance~\cite{Moghaddam2006SPCABounds, dAspremont2008PCAOptimalSolution, Birnbaum2013MinimaxSPCA}. Typically, the worst case is attained by the flat sparsity regime, where the nonzeros of \(\bfv\) have comparable magnitudes. Therefore, most algorithmic analyses are tailored to flat signals, and relatively little is known about how to leverage non-flat structures to improve the sample complexity.

This worst-case perspective overlooks a practical reality: signals are rarely ``flat'' sparse. Across genomics, imaging, and text, loadings often exhibit graded profiles (power-law or exponential) in which a few leading coefficients carry most of the energy~\cite{Carvalho2008SPCAGeneExpression,Seghouane2018AdaptiveBlockSPCA}. Intuitively, such concentration should reduce the sample complexity: since the prominent coordinates are easier to distinguish from noise, the total number of samples required for recovery ought to drop accordingly. Work on the closely related problem of sparse phase retrieval, whose initialization step can be formulated as an SPCA-type problem under a different observation model~\cite{Liu2021Towards}, has provided useful motivation for profile-aware spectral initialization that exploits such non-flat structure~\cite{Wu2021Hadamard,Cai2022Provable,Xu2024Exponential}. For SPCA itself, the single-peak analysis in~\cite{Cai25PeakPCA} shows that the sample complexity can drop to \(m\gtrsim k\log n\) when one coordinate carries a dominant share of the energy. Crucially, although these advances move beyond the flat-sparsity assumption, they characterize the energy profile only through its largest entry. As a result, although these methods apply beyond the single-peak setting, their theoretical guarantees remain controlled by the largest entry, making them insensitive to how the remaining energy is distributed. This motivates the central question of this paper:

\emph{Can we develop a polynomial-time SPCA algorithm whose guarantees are governed by the full energy profile?}

In this paper, we answer this question positively. We introduce \(s(p)\) to quantify the energy accumulation of the top \(p\) coordinates of the spike (see \Cref{def:structure-function} below for a formal definition), and propose Spectral Energy Pursuit (SEP). Conceptually, SEP gradually builds the support set by alternating between signal estimation on the restricted subset and coordinate selection on the full matrix. This allows the algorithm to exploit the cumulative energy across the top coordinates. Crucially, although no profile information is used, its sample complexity adapts to the energy profile: it matches the classic sample complexity in the flat regime and becomes strictly smaller as the signal energy becomes more concentrated; see \Cref{prop:power-law-interp} for the details.

\begin{table*}[t]
    \centering
    \caption{Sample-complexity landscape for SPCA}
    \label{tab:sc-landscape}
    \begin{tabular}{@{}cccc@{}}
        \toprule
        Line of work                                                                                                                                        & Sample complexity                      & Profile dependence & Notes                     \\
        \midrule
        Information-theoretic limits~\cite{VuLei2012MinimaxSPCA, WangLuLiu2014MinimaxSPCA, Johnstone2009SPCAConsistency, BerthetRigollet2013aSPCADetection} & \(k\log n\)                            & -                  & not poly-time             \\

        Diagonal Thresholding (DT)~\cite{Johnstone2009SPCAConsistency}                                                                                      & \(k^2\log n\)                          & none               & -                         \\
        Semidefinite Programming (SDP)~\cite{dAspremont2007PCASDP,Amini2009SPCASDP}                                                                         & \(k^2\log n\)                          & none               & poly-time but heavy       \\

        Single-peak-based methods~\cite{Cai25PeakPCA}                                                                                                       & \(ks(1)\log n\)                        & top-\(1\) only     & sensitive to seeding      \\
        \cmidrule(l{2cm}r{2cm}){1-4}\noalign{\vskip -0.4ex}
        \multicolumn{4}{c}{\emph{What is absent: polynomial-time guarantees depending on the full profile.}}                                                                                                                                          \\
        \cmidrule(l{2cm}r{2cm}){1-4}

        \textbf{SEP (ours)}                                                                                                                                 & \(\max_{1\le p\le k}\; ps^2(p)\log n\) & full profile       & \textbf{uniformly better} \\
        \bottomrule
    \end{tabular}
\end{table*}

We summarize the line of work on SPCA in \Cref{tab:sc-landscape}, where one can see that SEP is the first practical algorithm whose sample complexity fully adapts to the entire energy profile of the spike via the structure function \(s(p)\). Our contributions are threefold.

\begin{itemize}[itemsep=2pt, topsep=2pt]
    \item \textbf{Algorithmic contribution: Spectral Energy Pursuit (SEP).}
          We propose Spectral Energy Pursuit (SEP), a computationally efficient and profile-agnostic algorithm for SPCA. It admits a simple implementation and achieves robust performance in practice without requiring prior knowledge of the signal's energy structure.

    \item \textbf{Theoretical contribution: Instance-dependent sample complexity.}
          The sample complexity adapts to the full energy profile \(s(p)\), recovering \(k^2\log n\) for flat spikes and improving toward \(k\log n\) as energy concentrates, strictly outperforming single-peak-based guarantees.

    \item \textbf{Refinement contribution: Iteration-independent accuracy.}
          With a single centered truncated-power step, the estimator already reaches the statistical error floor; more iterations do not change the order.

\end{itemize}

Throughout the paper, for a vector \(\bfv\in\R^n\), \(\|\bfv\|_s\) denotes its \(\ell_s\) norm, and \(v_{(1)}\ge v_{(2)}\ge \cdots\) denote its sorted absolute entries. For the coordinate-selection operations used in our algorithms, \(\operatorname{Top}_q(\bfa)\) denotes the indices of the \(q\) largest entries of \(\bfa\), and \(\mcH_q(\bfa)\) retains precisely the entries indexed by \(\operatorname{Top}_q(|\bfa|)\) and sets the rest to zero. To make these selections deterministic, ties in \(\operatorname{Top}_q\) and every \(\arg\max\) are broken in favor of smaller coordinate indices.

For a matrix \(\bfA\in\R^{n\times n}\), \(\|\bfA\|_2\) is its spectral norm (i.e., largest singular value). For index sets \(S, U\subseteq[n]:=\{1,2,\dots,n\}\), \(\bfA_{S,U}\) is the submatrix of \(\bfA\) row-indexed by \(S\) and column-indexed by \(U\), and \(\bfv_S\) is the subvector of \(\bfv\) indexed by \(S\). For the samples \(\bfx_1,\dots,\bfx_m\in\R^n\), we use \(\bfx_i(S)\) to denote the subvector of \(\bfx_i\) indexed by \(S\). For two positive sequences \(a_n,b_n>0\), we write \(a_n \lesssim b_n\) or \(b_n \gtrsim a_n\) if there exists an absolute constant \(C>0\) such that \(a_n \le C b_n\) for all sufficiently large \(n\); we write \(a_n \asymp b_n\) if both \(a_n \lesssim b_n\) and \(b_n \lesssim a_n\) hold. Absolute constants \(C,c>0\) (possibly with subscripts) are allowed to vary between occurrences.
Unless specified, they are universal (independent of problem parameters), and this convention is in force throughout statements and proofs.

The remainder of this paper is organized as follows. \Cref{sec:algorithm} reviews existing algorithms and then presents SEP and key intuitions. \Cref{sec:theory} states the main theoretical results, followed by proofs in \Cref{sec:proof}. \Cref{sec:discussion} explains the structural gain of SEP, relates \(s(p)\) to weak sparsity and \(\ell_q\)-type descriptions, addresses data dependence and the operator choice in refinement, and discusses limitations including the dependence on \(\theta\). \Cref{sec:simulations} provides numerical experiments, and \Cref{sec:conclusion} concludes the paper. Appendices collect some auxiliary lemmas and additional proofs.

%% file: 03Algorithm.tex
\section{Algorithm}\label{sec:algorithm}

To better understand the design of our Spectral Energy Pursuit (SEP) algorithm for sparse PCA (SPCA), we first revisit two classical approaches and their intuitions: (i) diagonal-thresholding methods and (ii) single-peak-based methods that leverage the magnitude of the largest nonzero entry. Then we present our SEP algorithm and explain why it is effective. Finally, we present a refinement technique using TPower to further improve the estimate from SEP without increasing the sample complexity. Here we briefly recall the model setup used throughout the paper.

We consider the standard spiked covariance model, where the population covariance takes on the form
\begin{equation}
    \bSigma = \bfI_n + \theta\, \bfv\bfv^\top,
\end{equation}
and the sample covariance \(\hSigma\) is computed by
\begin{equation}
    \hSigma = \frac1m \sum_{i=1}^m \bfx_i \bfx_i^\top,
\end{equation}
where \(\{\bfx_i\}_{i=1}^m\) are \iid samples drawn from \(\mathcal{N}(\mathbf{0}, \bSigma)\). Throughout this paper, we denote
\begin{equation}\label{eq:hGamma-def}
    \hGamma := \hSigma - \bfI_n,
\end{equation}
where \(\hSigma\) is the sample covariance matrix and \(\bfI_n\) is the \(n\times n\) identity matrix.
We further decompose \(\hGamma\) by a noise component \(\bfW\) as
\begin{equation}\label{eq:matrix-decomp}
    \bfW := \hGamma - \theta\, \bfv\bfv^\top.
\end{equation}

To streamline intuition, in this section we reason in a regime where the number of samples is large enough so that the sample perturbation \(\bfW\) is small relative to the signal. The formal theory in \Cref{sec:theory} provides uniform quantitative operator-norm bounds on principal blocks of \(\bfW\) together with the resulting sample complexity and error rates.

\subsection{Classical SPCA Algorithms and Principles}\label{sec:classical}

\subsubsection{Diagonal Thresholding (DT) Methods}\label{sec:diag-thresh}
Diagonal-thresholding~\cite{Ma2013SPCAIT, Amini2009SPCASDP} (and closely related covariance-thresholding~\cite{Paul2007SPCAeigenstructure}) methods estimate a support by screening coordinates with large empirical variances or norms of the sample covariance. A basic selector takes the \(k\) coordinates with largest \(\hGamma_{jj}\) (or \(\hSigma_{jj}\), the rationale is the same) and then computes the top eigenvector on this restricted submatrix. The intuition is that, because \(\EE[\hGamma_{jj}] = \theta v_j^2\) for any \(j\), when coordinate \(j\) carries significant spike energy, \(\hGamma_{jj}\) is positively biased by \(\theta v_j^2\), making it stand out after concentration.

For clarity, we state bounds in the exactly \(k\)-sparse case, so \(v_{(k)}>0\) and \(v_{(k)}=\min_{j\in\supp(v)}|v_j|\). The sample size required to recover all support coordinates via diagonal screening obeys (up to factors depending on \(\theta\)) \cite{Ma2013SPCAIT, Amini2009SPCASDP}
\begin{equation}\label{eq:DTsample}
    m\  \gtrsim\  \frac{k}{v_{(k)}^2} \log n.
\end{equation}
We emphasize that this bound~\eqref{eq:DTsample} is tailored to exact support recovery. Consequently, the requirement worsens when the weakest nonzero coordinate is very small. A common way to mitigate this is to impose an additional lower bound on the energy in the support, e.g., \(v_{(k)}^2 \gtrsim 1/k\), which essentially implies \(v_{(k)}^2 \asymp 1/k\), yielding the familiar uniform sufficient scaling \(m \gtrsim k^2\log n\).

\input{figs/fig_DT.tex}

\Cref{fig:dt-matrix} illustrates the core idea of DT. The diagonal entries in the support are elevated by the spike energy \(\theta v_{(i)}^2\), allowing them to be separated from the diagonals outside the support after concentration.

\subsubsection{Single-Peak-Based Methods (Largest-Entry Energy)}\label{sec:peakiness}

Recent analyses~\cite{Cai25PeakPCA} show improved performance when the spike exhibits a dominating entry (single peak).
Roughly, one can tie the sample complexity to \(v_{(1)}\), the energy of the largest coordinate, and provably outperform flat-signal guarantees when \(v_{(1)}\) is sufficiently large.
A representative procedure is as follows:
(i) identify \(j_{\max}=\arg\max_j \hGamma_{jj}\) by diagonal screening;
(ii) use the max-column proxy \(\hGamma_{\cdot, j_{\max}}\) to score coordinates; and
(iii) select the top-\(k\) entries of \(|\hGamma_{\cdot, j_{\max}}|\) to form the support estimate \(S\), then compute the top eigenvector on \(\hGamma_{S,S}\).
It holds that
\begin{equation}\label{eq:peakiness-proxy}
    \hGamma_{\cdot, j_{\max}} \;=\; \theta\, v_{j_{\max}} \bfv + \bfW_{\cdot, j_{\max}}
    \;\approx\; \theta\, v_{(1)} \bfv ,
\end{equation}
since the screener typically picks an index attaining a top entry of \(\bfv\), leading to \(v_{j_{\max}} \approx v_{(1)}\).
Therefore, the proxy \(\hGamma_{\cdot, j_{\max}}\) is proportional to \(\bfv\) and scaled by \(v_{(1)}\).
When \(v_{(1)}^2 \gg 1/k\) (non-flat spikes), this common multiplicative boost sharpens the separation between coordinates in the support and those outside it under concentration, leading to the scaling (up to \(\theta\)-dependent factors)
\begin{equation}\label{eq:peakiness-sample}
    m\ \gtrsim\ \frac{k}{v_{(1)}^2}\log n.
\end{equation}

\input{figs/fig_Peakiness.tex}

\Cref{fig:peakiness-matrix} illustrates the single-peak-based method. The sample complexity~\eqref{eq:peakiness-sample} is much smaller than \(k^2\log n\) when the signal is highly spiky, e.g., it becomes \(k\log n\) given \(v_{(1)}^2 \asymp 1\). In contrast, the advantage disappears and \eqref{eq:peakiness-sample} matches the \(k^2\log n\) order in the flat regime \(v_{(1)}^2 \asymp 1/k\). The gain comes from using a cross-coordinate proxy built from the column of the (estimated) largest entry: the single-peak heuristic uses \(\hGamma_{i, j_{\max}} \approx \theta\, v_{(1)} v_i\), whereas DT only uses the diagonal entry \(\hGamma_{ii} \approx \theta v_i^2\). This cross-coordinate amplification particularly helps non-flat profiles, where many \(v_i\) are small but get boosted by the leading factor
\(v_{(1)}\), whereas the diagonal statistic \(v_i^2\) remains too weak to pass the screening threshold.

\subsection{Spectral Energy Pursuit (SEP)}\label{sec:SEP}
Single-peak-based approaches exploit the largest entry \(v_{(1)}\) to bootstrap support recovery. While effective when \(v_{(1)}^2 \gg 1/k\), they face two structural limitations:
\begin{enumerate}[itemsep=2pt, topsep=2pt]
    \item \textbf{Single-anchor perspective:} guarantees are typically anchored to the largest coordinate, which may underutilize the cumulative energy spread across multiple top entries when the spike is less pronounced.
    \item \textbf{Sensitivity to the seeding step:} it first identifies \(j_{\max}=\arg\max_j \hGamma_{jj}\) and then builds a column proxy around it; this can make performance sensitive to the initial screener.
\end{enumerate}

\begin{algorithm}[ht]
    \caption{\textsc{Spectral Energy Pursuit (SEP)}}
    \label{alg:SEP}
    \begin{algorithmic}[1]
        \Require Samples \(\{\bfx_i\}_{i=1}^m\), sparsity budget \(k\).
        \State \(\hSigma \leftarrow \frac1m\sum_i \bfx_i \bfx_i^\top\), \(\hGamma\leftarrow \hSigma-\bfI\), \(d_j\leftarrow \hGamma_{jj}\).
        \State \(S^{(1)}\leftarrow \{\arg\max_j |d_j|\}\).
        \For{\(p=1\) to \(k-1\)}
        \State \(\hat\bfe^{(p)}\leftarrow \text{top-eigvec of }\hGamma_{S^{(p)},S^{(p)}}\); zero-pad \(\hat\bfe^{(p)}\) to \(\R^n\) and normalize.
        \State \(\bfu^{(p)}\leftarrow \hGamma  \hat\bfe^{(p)}\).
        \State \(S^{(p+1)}\leftarrow \operatorname{Top}_{p+1}(|\bfu^{(p)}|)\). \Comment{reselection}
        \EndFor
        \State \(\hat \bfv\leftarrow\) top-eigvec of \(\hGamma_{S^{(k)},S^{(k)}}\); zero-pad \(\hat \bfv\) to \(\R^n\) and normalize.
        \Ensure \(\hat \bfv\).
    \end{algorithmic}
\end{algorithm}

\input{figs/fig_SEP.tex}

Motivated by these considerations, we present SEP (see~\Cref{alg:SEP}) and illustrate its iterative reselection mechanism in \Cref{fig:SEP-matrix}.
Similar to single-peak methods, SEP starts from diagonal screening to pick the first coordinate, but then proceeds in \(k-1\) rounds of eigenvector computation and reselection to gradually build up the support.
Given the current support \(S^{(p)}\) (where \(p \in [k-1]\)), let \(\hat\bfe^{(p)}\) be the top eigenvector of the restricted submatrix \(\hGamma_{S^{(p)},S^{(p)}}\).
The response decomposes as
\begin{equation}\label{eq:SEP-response}
    \hGamma \hat\bfe^{(p)} \;=\; \theta\,\langle \bfv,\hat\bfe^{(p)}\rangle\, \bfv \;+\; \bfW \hat\bfe^{(p)},
\end{equation}
separating a signal term, whose magnitude scales with the current alignment \(|\langle \bfv,\hat\bfe^{(p)}\rangle|\), from a noise term bounded by concentration.
Intuitively, if \(\hat\bfe^{(p)}\) is reasonably aligned with \(\bfv\), the signal term lifts high-energy coordinates; selecting the top-\((p{+}1)\) entries increases the total spike energy on \(S^{(p+1)}\), which in turn improves the alignment of the next eigenvector.
This creates a positive feedback: a better alignment yields a cleaner ranking, leading to more captured energy.
Crucially, under the high-probability event of bounded noise (see \Cref{prop:principal-submatrix}), this loop is stable: it tolerates intermediate selection errors (e.g., local swaps) without diverging, ensuring the estimate progressively improves rather than degrades.

We note that this mechanism coincides with classical heuristics for small sparsity: SEP reduces to diagonal thresholding for \(k=1\) and the single-peak method for \(k=2\). However, for \(k > 2\), a distinct mechanism emerges: unlike static heuristics relying on fixed anchors, SEP utilizes the aforementioned iterative spectral feedback loop.
Crucially, in contrast to DT and ``peakiness'' methods that hinge on per-coordinate margins for separation and signal estimation, SEP adopts a cumulative-energy viewpoint.
The reason is that, when adjacent magnitudes are nearly tied (\(v_{(p+1)}\!\approx\! v_{(p)}\)), enforcing a strict entry-wise ordering requires resolving vanishingly small differences, which drastically inflates the necessary sample size.
SEP instead tolerates local swaps across consecutive rounds: it may temporarily include the \((p{+}1)\)-st index before the \(p\)-th without harming estimation, as long as the selected set retains sufficient total energy.

We formalize this requirement via an energy-lower-bound invariant.
At each round \(p\in [k-1]\), we require that the selected support preserves a fixed fraction of the cumulative spike energy:
\begin{equation}\label{eq:SEP-energy-invariant}
    \|\bfv_{S^{(p)}}\|_2^2 \ge \gamma \sum_{i=1}^p v_{(i)}^2,
\end{equation}
where \(\gamma\in(0,1)\) is a constant.
In terms of the structure function \Cref{def:structure-function}, condition~\eqref{eq:SEP-energy-invariant} reads \(\|\bfv_{S^{(p)}}\|_2 \ge \sqrt{\gamma/s(p)}\).
Iterating from \(p = 1\) to \(k-1\) yields \(\|\bfv_{S^{(k)}}\|_2 \ge \sqrt{\gamma}\) and the final angle bound, which underpins \Cref{thm:main}.

Finally, we briefly analyze the computational cost of SEP, separating preprocessing from the subsequent iterations. Explicitly forming the dense sample covariance \(\hSigma\) (and hence the centered covariance \(\hGamma\)) from \(m\) length-\(n\) samples costs \(\mcO(mn^2)\) time and requires \(\mcO(n^2)\) storage. Once \(\hGamma\) has been formed, the initial diagonal screening costs \(\mcO(n)\). At round \(p\), a standard eigensolve on the \(p\times p\) principal submatrix costs up to \(\mcO(p^3)\), while forming the length-\(n\) response \(\bfu^{(p)}=\hGamma\hat\bfe^{(p)}\) costs \(\mcO(np)\) because \(\hat\bfe^{(p)}\) is \(p\)-sparse. Selecting the top \(p{+}1\) entries from this response costs \(\mcO(n)\) with a linear-time selection algorithm, or \(\mcO(n\log n)\) by sorting. With linear-time selection, the aggregate selection cost \(\mcO(nk)\) is absorbed, and the subsequent SEP iterations, including the final restricted eigensolve, cost
    \begin{equation}
        \sum_{p=1}^{k-1}\mcO\bigl(np+p^3\bigr)+\mcO(k^3)
        =\mcO\bigl(nk^2+k^4\bigr).
    \end{equation}
    If sorting is used instead, an additional \(\mcO(nk\log n)\) term accounts for all reselection steps. Thus, with linear-time selection, preprocessing and subsequent iterations together cost \(\mcO(mn^2+nk^2+k^4)\) time, while the explicitly stored dense covariance requires \(\mcO(n^2)\) space.

\subsection{Post-refinement with TPower}\label{sec:post-tpower}

\begin{algorithm}[ht]
    \caption{\textsc{TPower Post-Refinement}}
    \label{alg:tpower-refine}
    \begin{algorithmic}[1]
        \Require Sample covariance operator \(\hGamma\), sparsity \(k\), SEP output \(\hat\bfv\) (unit norm),
        iterations \(T\), parameter \(k' \ge k\).
        \State Initialize \(\bfw^{(0)} \leftarrow  \hat\bfv\).
        \For{\(t=0\) to \(T-1\)}
        \State \emph{Multiplication:} \(\bfy \leftarrow \hGamma \bfw^{(t)}\).
        \State \emph{Truncation:} \(\bfz \leftarrow \mcH_{k'}(\bfy)\) \Comment{retain \(\operatorname{Top}_{k'}(|\bfy|)\)}
        \State \emph{Normalization:} \(\bfw^{(t+1)} \leftarrow \bfz/\|\bfz\|_2\).
        \EndFor
        \State \textbf{Output:} \(\hat \bfv_{\mathrm{refine}} \leftarrow \bfw^{(T)}\).
    \end{algorithmic}
\end{algorithm}

To further improve estimation accuracy, we apply the truncated power method (TPower) as a post-refinement to the SEP output. TPower, introduced by Zhang and collaborators~\cite{Yuan2013TPower} and now widely used in sparse spectral estimation~\cite{Cai2013OptimalSPCA}, alternates a spectral update with hard thresholding. In our analysis (see \Cref{sec:theory-tpower}), running a single iteration with the centered operator \(\hGamma\) already attains the statistical error bound, and extra iterations do not change the order of the statistical error. Practically, we therefore use one or a few iterations as a lightweight polish; \Cref{alg:tpower-refine} gives a fast implementation. The choice of operator is important: using the raw covariance \(\hSigma\) introduces a carry-over term in the spectral update that leaves an optimization residual across iterations, whereas the centered operator \(\hGamma\) avoids this effect and underlies the one-iteration phenomenon. This is discussed in \Cref{sec:dis-tpower}.

From a computational perspective, one TPower iteration costs a matrix-vector multiply with \(\hGamma\) whose cost is \(\mcO(n^2)\), plus a top-\(k'\) selection, whose cost is \(\mcO(n)\) with a selection algorithm or \(\mcO(n\log n)\) by sorting. Thus, \(T\) iterations cost \(\mcO(T n^2)\) time, which is polynomial-time and efficient for moderate \(n\).

%% file: figs/fig_DT.tex
\begin{figure}[ht]
    \centering
    \def\cs{0.35}
    \begin{tikzpicture}[scale=1]
        \foreach \i in {0,1,2,3,4,5,6,7}{
                \foreach \j in {0,1,2,3,4,5,6,7}{
                        \draw (\i*\cs,-\j*\cs) rectangle ++(\cs,-\cs);
                    }
                \filldraw[fill=gray,draw=black] (\i*\cs,-\i*\cs) rectangle ++(\cs,-\cs);
            }
        \node[anchor=west,font=\scriptsize] at (8*\cs+0.2, -1.0) {Select top-\(k\) of \(\hGamma_{jj}\)};

        \node[anchor=west,font=\scriptsize] at (8*\cs+0.2, -2.0) {\textbf{Successful margin:} \(\theta v_{(k)}^2\).};

    \end{tikzpicture}
    \caption{Diagonal Thresholding algorithm: selects the support indices based on the largest diagonal entries of the centered covariance shift matrix \(\hGamma\). The success condition depends on the smallest nonzero entry \(v_{(k)}\).}
    \label{fig:dt-matrix}
\end{figure}
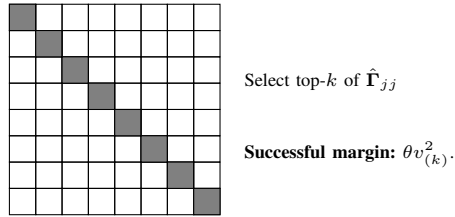

%% file: figs/fig_Peakiness.tex
\begin{figure}[ht]
    \centering
    \def\cs{0.35}
    \begin{tikzpicture}[scale=1]
        \def\jmax{2}
        \fill[cyan!80] (\jmax*\cs, 0) rectangle ++(\cs, -8*\cs);
        \foreach \i in {0,1,2,3,4,5,6,7}{
                \foreach \j in {0,1,2,3,4,5,6,7}{
                        \draw (\i*\cs,-\j*\cs) rectangle ++(\cs,-\cs);
                    }
                \filldraw[fill=gray,draw=black] (\i*\cs,-\i*\cs) rectangle ++(\cs,-\cs);
            }
        \filldraw[fill=red!80,draw=black] (\jmax*\cs,-\jmax*\cs) rectangle ++(\cs,-\cs);
        \node[anchor=west,font=\scriptsize] at (8*\cs+0.2, -0.6) {Select \(j_{\max}=\arg\max_j \hGamma_{jj}\)};
        \node[anchor=west,font=\scriptsize] at (8*\cs+0.2, -1.2) {Use column \(|\hGamma_{\cdot,j_{\max}}|\) to select top-\(k\)};
        \node[anchor=west,font=\scriptsize] at (8*\cs+0.2, -1.8) {\textbf{Successful margin:} \(\theta v_{(1)}^2 \to\theta v_{(1)}v_{(k)}\).};
    \end{tikzpicture}
    \caption{Single-peak-based algorithm: selects the support indices based on the column of the centered covariance shift matrix \(\hGamma\) corresponding to the largest diagonal entry. The success condition depends on the largest and smallest nonzero entries \(v_{(1)}\) and \(v_{(k)}\).}
    \label{fig:peakiness-matrix}
\end{figure}
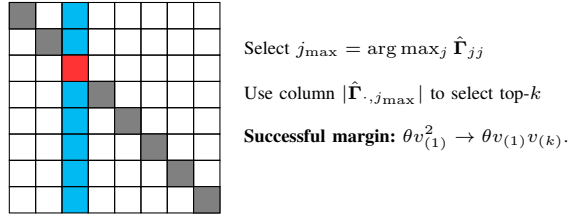

%% file: figs/fig_SEP.tex
\begin{figure}[ht]
    \centering
    \def\cs{0.35}
    \resizebox{\columnwidth}{!}{%
        \begin{tikzpicture}[scale=1]
            \foreach \i in {2,3,4}{
                    \foreach \j in {2,3,4}{
                            \fill[orange!80] (\i*\cs,-\j*\cs) rectangle ++(\cs,-\cs);
                        }
                }
            \foreach \i in {0,1,2,3,4,5,6,7}{
                    \foreach \j in {0,1,2,3,4,5,6,7}{
                            \draw (\i*\cs,-\j*\cs) rectangle ++(\cs,-\cs);
                        }
                }
            \node[font=\scriptsize] at (4*\cs, 0.35) {\(S^{(p)}\!\times\! S^{(p)}\)};

            \begin{scope}[xshift=3.7cm]
                \def\vw{0.35} \def\vh{\cs} 
                \foreach \j in {0,1,2,3,4,5,6,7}{
                        \draw (0,-\j*\vh) rectangle ++(\vw,-\vh);
                    }
                \foreach \r in {1,3,4,6}{
                        \filldraw[fill=red!80, draw=black] (0,-\r*\vh) rectangle ++(\vw,-\vh);
                    }
                \node[font=\scriptsize,anchor=west] at (-0.6, 0.35) {\(\bfu^{(p)}=\hGamma \hat\bfe^{(p)}\)};
            \end{scope}
            \draw[-{Latex[length=2mm]}] (8.5*\cs, -4*\cs) -- ++(0.5, 0);
            \draw[-{Latex[length=2mm]}] (9.5*\cs+1.0, -4*\cs) -- ++(0.5, 0);
            \begin{scope}[xshift=5.0cm]
                \foreach \i in {1,3,4,6}{
                        \foreach \j in {1,3,4,6}{
                                \fill[cyan!80] (\i*\cs,-\j*\cs) rectangle ++(\cs,-\cs);
                            }
                    }
                \foreach \i in {0,1,2,3,4,5,6,7}{
                        \foreach \j in {0,1,2,3,4,5,6,7}{
                                \draw (\i*\cs,-\j*\cs) rectangle ++(\cs,-\cs);
                            }
                    }
                \node[font=\scriptsize] at (4*\cs, 0.35) {\(S^{(p+1)}\!\times\! S^{(p+1)}\)};
                \node[font=\footnotesize, align=center, text width=7.6cm, inner sep=0pt]
                at (-0.9, -8*\cs-0.72) {\textbf{Estimated energy lower bound \(\|\bfv_{S^{(\cdot)}}\|_2\):}\\[-1pt]
                    \(\sqrt{\gamma/s(1)} \to \dots \to \sqrt{\gamma/s(p)} \to \dots \to \sqrt{\gamma/s(k)}\)};
            \end{scope}
        \end{tikzpicture}%
    }

    \caption{Spectral Energy Pursuit algorithm: at each round \(p\), SEP forms the vector \(\bfu^{(p)}=\hGamma \hat\bfe^{(p)}\) by multiplying the centered covariance shift matrix \(\hGamma\) with the vector \(\hat\bfe^{(p)}\), the top eigenvector of \(\hGamma_{S^{(p)}, S^{(p)}}\). The next support estimate \(S^{(p+1)}\) is obtained by selecting the top-\((p+1)\) entries of \(\bfu^{(p)}\). The success condition depends on the signal energy structure function \(s(p)\).}
    \label{fig:SEP-matrix}
\end{figure}

%% file: 04Theory.tex
\section{Main Results}\label{sec:theory}
\subsection{Preliminary}\label{sec:theory-preliminary}
Before stating our main theorem, we formally define the signal-energy structure function \(s(p)\).
\begin{definition}
    [Signal-energy structure function]\label{def:structure-function}
    Given a unit spike vector \(\bfv\in\R^n\) with sparsity \(k\), define its signal-energy structure function \(s(p)\) for \(1\le p\le k\) as follows
    \begin{equation}
        s(p) := \Big(\sum_{i=1}^p v_{(i)}^2\Big)^{-1}.
    \end{equation}
\end{definition}
The function \(s(p)\) captures the energy accumulation of the top \(p\) coordinates of the spike. It is easy to see that \(s(k)=1\) always holds, and \(s(1)\) ranges from \(1\) (all energy concentrated on one entry) to \(k\) (flat spike with equal energy on all support entries). The ordering of the squared magnitudes gives \(p v_{(p+1)}^2\le \sum_{i=1}^p v_{(i)}^2=1/s(p)\), and hence, for \(1\le p<k\),
    \begin{equation}\label{eq:structure-ratio}
        \frac{s(p)}{s(p+1)}
        = 1+s(p)v_{(p+1)}^2
        \le 1+\frac{1}{p}
        = \frac{p+1}{p}.
    \end{equation}
    By definition, \(s(p)\) is non-increasing because its reciprocal is a non-decreasing partial sum; moreover, inequality~\eqref{eq:structure-ratio} is equivalent to \(p s(p)\le (p+1)s(p+1)\), so \(p s(p)\) is non-decreasing.

Next, we define the error metric used in our analysis, which measures the sine of the angle between the estimated direction \(\hat \bfv\) and the true spike \(\bfv\).
\begin{definition}
    [Direction metric]
    Given two unit vectors \(\hat \bfv, \bfv\in\R^n\),
    the cosine of the angle between them is defined by
    \begin{equation}
        \cos\angle(\hat \bfv, \bfv) := |\hat \bfv^\top \bfv|.
    \end{equation}
    Their direction error is defined by
    \begin{equation}\label{eq:sin-distance}
        \sin\angle(\hat \bfv, \bfv) := \sqrt{1-\cos^2\angle(\hat \bfv, \bfv)} = \sqrt{1 - |\hat \bfv^\top \bfv|^2}.
    \end{equation}
\end{definition}

Finally, we introduce a high-probability event \(\mcE\) that will be used throughout our analysis. Recall that \(\bfW = \hGamma - \theta\, \bfv\bfv^\top\) is the noise matrix defined in \eqref{eq:matrix-decomp}. We define
\begin{equation}\label{eq:uniform-event}
    \mcE \;:=\; \bigcap_{p=1}^{n}\ \bigcap_{\substack{S\subset [n]\\ |S|=p}}
    \left\{\,\|\bfW_{S,S}\|_2 \le C(1+\theta)\sqrt{\frac{p\log n}{m}}\,\right\},
\end{equation}
for some absolute constant \(C>0\). In plain words, this event establishes a uniform spectral upper bound on the noise fluctuations across all possible principal submatrices. It guarantees that for any subset of coordinates the algorithm might visit, the noise energy remains strictly bounded relative to the subset size. This uniformity is crucial for handling the data-dependent nature of the support selection, as it rules out the existence of any ``worst-case'' blocks that could mislead the algorithm. The probability of the event \(\mcE\) is controlled by the following proposition.

\begin{proposition}[Principal-submatrix spectral bound]\label{prop:principal-submatrix}
    There exist absolute constants \(C,c>0\) such that, with probability at least \(1-n^{-c}\), it holds that
    \begin{equation}
        \|\bfW_{S,S}\|_{2}
        \;\le\; C(1+\theta)\sqrt{\frac{p\,\log n}{m}},
    \end{equation}
    for every \(p\in[n]\) and every index set \(S\subset[n]\) with \(|S|=p\). This implies \(\PP(\mcE)\ge 1-n^{-c}\).
\end{proposition}

We condition on \(\mcE\) throughout. It provides uniform, path-wise spectral control for all data-dependent supports and absorbs all probabilistic statements up front, so the remainder of the analysis is deterministic. Moreover, for any (data-dependent) index sets \(S_1,S_2\subset[n]\) with \(|S_1|,|S_2|\le p\), letting \(S:=S_1\cup S_2\) yields the embedding \(\|\bfW_{S_1,S_2}\|_2 \le \|\bfW_{S,S}\|_2\), so \eqref{eq:uniform-event} also controls rectangular blocks used by reselection (up to a benign \(\sqrt{2}\) factor). Working on \(\mcE\) substantially streamlines the analysis, and can avoid some technical complications that require careful union bounds over data-dependent supports; see \Cref{sec:data-dependence} for more details.

\subsection{Results of SEP}\label{sec:theory-SEP}
\Cref{thm:main} below is our main result. It states that SEP enjoys a structure-adaptive sample complexity depending on the function \(s(p)\).

\begin{theorem}[Profile-adaptive sample complexity for direction estimation] \label{thm:main}
    Condition on the high-probability event \(\mcE\). For any \(\gamma\in(0,1)\), if
    \begin{equation}\label{eq:m-bound-gamma}
        m \ \ge\ C_1 \frac{(1+\theta)^2}{\theta^2 \gamma^2(1-\sqrt{\gamma})^2} \max_{1\le p \le k}p s^2(p)\log n,
    \end{equation}
    then
    the final selected support \(S^{(k)}\) in~\Cref{alg:SEP} satisfies \(\|\bfv_{S^{(k)}}\|_2\ge \sqrt{\gamma}\), and
    \begin{equation}\label{eq:SEP-error}
        \sin\angle(\hat \bfv,\bfv)\ \le\ \underbrace{\sqrt{1-\gamma}}_{\text{approximation error}}
        + \quad
        \underbrace{\frac{C_2(1+\theta)}{\theta \gamma}\sqrt{\frac{k\log n}{m}}}_{\text{statistical error}},
    \end{equation}
    where \(\hat \bfv\) is the output of \Cref{alg:SEP}.
\end{theorem}

On the one hand, the sample complexity bound~\eqref{eq:m-bound-gamma} matches the best known order for many practical algorithms in the flat case, and strictly improves when the signal is concentrated. Specifically, in the flat case \(v_{(i)}^2 = 1/k\), we have \(s(p)=k/p\), so \(\max_p p s^2(p)=k^2\) and condition \eqref{eq:m-bound-gamma} reduces to \(m\gtrsim k^2\,\log n\). In contrast, for highly concentrated profiles where \(s(p)\asymp 1, \forall 1\le p\le k\), it yields scaling \(m \asymp k\log n\). For intermediate profiles, the dependence on \(s(p)\) yields the sample complexity that varies smoothly between these extremes; see \Cref{prop:power-law-interp} for an explicit continuum. Compared with prior polynomial-time guarantees, this scale~\eqref{eq:m-bound-gamma} is never larger and is strictly smaller on broad non-flat classes; we show this in \Cref{thm:superiority}.

On the other hand, the error bound \eqref{eq:SEP-error} cleanly separates two effects. The first term is an approximation error: since the selected support retains a \(\gamma\)-fraction of the spike energy, there is an intrinsic angular bound of order \(\sqrt{1-\gamma}\). The second term is statistical error: it scales as \(\sqrt{(k\log n)/m}\) (up to constants), vanishes as \(m\to\infty\), and matches information-theoretic lower bounds in the worst-case (flat) regime~\cite{VuLei2012MinimaxSPCA,WangLuLiu2014MinimaxSPCA}, hence is minimax-rate optimal there. Moreover, the impact of the signal shape is mediated entirely by the achievable support-energy level \(\gamma\), and the specific value of \(\gamma\) plays a dual role while the asymptotic statistical rate remains invariant. Specifically, at a fixed sample size \(m\), more concentrated profiles (reflected in a smaller \(\max_{p\le k} p\,s^2(p)\) in \eqref{eq:m-bound-gamma}) permit a larger \(\gamma\) (closer to \(1\)). This simultaneously suppresses the intrinsic approximation error \(\sqrt{1-\gamma}\) and improves the prefactor \(1/\gamma\) of the statistical term.

To elucidate why the sample complexity order \(\max_{1\le p \le k} p s^2(p)\) arises, we consider an intermediate support size \(p\). The current signal energy captured is \(\|\bfv_{S^{(p)}}\|_2 \asymp \sqrt{1/s(p)}\) (\Cref{prop:energy-preservation}),
while the spectral noise on \(p\times p\) principal blocks concentrates at
\(\|\bfW_{S^{(p)},S^{(p)}}\|_2 \lesssim \sqrt{p\log n/m}\).
Suppressing constants and the dependence on \(\theta\), a Davis--Kahan step (\Cref{lem:DK}) in the perturbative regime yields
    \begin{align}
        |\langle \bfv,\hat \bfe^{(p)}\rangle|
         & \ \gtrsim\
        \|\bfv_{S^{(p)}}\|_2
        \sqrt{1 - \frac{\|\bfW_{S^{(p)},S^{(p)}}\|_2^2}{\|\bfv_{S^{(p)}}\|_2^4}} \nonumber \\
         & \ \gtrsim\
        \sqrt{\frac{1}{s(p)}}\sqrt{1 - \frac{p\,s^2(p)\,\log n}{m}}.\label{eq:intuition-aligment-bound}
    \end{align}
    In the next update,
    \(\hGamma \hat\bfe^{(p)}=\theta\, \langle \bfv,\hat\bfe^{(p)}\rangle \bfv + \bfW\hat\bfe^{(p)}\),
    and the reselection step (\Cref{lem:reselect}) ensures
    \begin{equation}\label{eq:intuition-reselection}
        \|\bfv_{S^{(p+1)}}\|_2 \ \ge\ \sqrt{\frac{1}{s(p+1)}} \;-\; \frac{C(1+\theta)}{\theta\,|\langle \bfv,\hat\bfe^{(p)}\rangle|}\sqrt{\frac{(p+1)\log n}{m}}.
    \end{equation}
    Plugging the alignment bound~\eqref{eq:intuition-aligment-bound} into~\eqref{eq:intuition-reselection} shows that a single step increases the captured energy from order \(1/s(p)\) to \(1/s(p{+}1)\) whenever
    \begin{equation}
        m \ \gtrsim\ p\,s^2(p)\,\log n .
    \end{equation}
    Because the algorithm deterministically grows the support from \(p=1\) to \(p=k\), we guarantee that every intermediate step succeeds by taking the path-wise maximum
    \begin{equation}
        m \ \gtrsim\ \max_{1\le p\le k} p\,s^2(p)\log n.
    \end{equation}
At the end (on the \(k\times k\) block), a final eigenvector estimation contributes the usual statistical error term
\(\lesssim \sqrt{\frac{k\log n}{m}}\) in~\eqref{eq:SEP-error}.

This heuristic explains the profile-dependent scaling. To visualize how \(\max_p p\,s^2(p)\) interpolates between flat and concentrated regimes, consider a power-law profile as shown in \Cref{prop:power-law-interp}.
\begin{proposition}[Power-law signal profiles: interpolation between flat and concentrated regimes]\label{prop:power-law-interp}
    Let \(v_{(i)}^2 = \lambda \cdot i^{-\alpha}\) for \(i=1,\ldots,k\), where \(\lambda=\left(\sum_{i=1}^k i^{-\alpha}\right)^{-1}\) so that \(\sum_{i=1}^k v_{(i)}^2=1\). Then
    \begin{equation}
        \max_{1\le p\le k} \; p\, s^2(p) \;\asymp\;
        \begin{cases}
            k^{\,2-2\alpha}, & 0 \le \alpha < \tfrac{1}{2}, \\[4pt]
            k,               & \alpha \ge \tfrac{1}{2}.
        \end{cases}
    \end{equation}
    Consequently, the sample complexity \(m \gtrsim \max_p p s^2(p)\log n\) interpolates from \(k^2\log n\) at \(\alpha=0\) (flat) to \(k\log n\) for
    \(\alpha \ge \tfrac12\) (concentrated).
\end{proposition}
Recent single-peak-based analyses (e.g.,~\cite{Cai25PeakPCA}) achieve the \(k\log n\) rate only under a strong dominance assumption, effectively requiring \(s(1)\asymp 1\) (an overwhelming leading entry). For the power-law family in \Cref{prop:power-law-interp}, however, our criterion \(m\gtrsim \max_{p\le k} p\,s^{2}(p)\log n\) already yields \(k\log n\) for all \(\alpha\ge\frac12\).
For example, in the case \(\alpha=\frac12\), it can be shown that \(s(1)\asymp \sqrt{k}\) (see \eqref{eq:power-law-sp}), so single-peak driven bounds inflate to \(k^{3/2}\log n\). This reflects a structural advantage beyond the largest coordinate. The rigorous results are given in \Cref{sec:superiority}, where we show that SEP attains strictly better sample complexity on certain non-flat profiles while never worsening the order relative to existing guarantees for all profiles.

\subsection{Results of TPower}\label{sec:theory-tpower}
When applying the TPower refinement after SEP, the approximation error \(\sqrt{1-\gamma}\) can be eliminated and only statistical error remains, as stated in the following theorem.

\begin{theorem}[TPower after \(T\) iterations: uniform statistical upper bound]\label{thm:tpower-T}
    Let the initialization \(\bfw^{(0)}\) of \Cref{alg:tpower-refine} be the output of \Cref{alg:SEP} whose support \(S^{(k)}\) satisfies \(\|\bfv_{S^{(k)}}\|_2\ge \sqrt{\gamma}\) for some \(\gamma\in(0,1)\). Let \(\bfw^{(T)}\) be the \(T\)-iteration output of \Cref{alg:tpower-refine} with keep-\(k'\) thresholding. Condition on the high-probability event \(\mcE\).
    When
    \begin{equation}\label{eq:m-bound-tpower}
        m \ \ge\ C_1\,\frac{(1+\theta)^2}{\theta^2\gamma^2}\,k'\log n,
    \end{equation}
    it holds that
    \begin{equation}
        \sin\angle\bigl(\bfw^{(T)},\bfv\bigr)\ \le\ C_2\,\frac{1+\theta}{\theta \gamma}\,\sqrt{\frac{k'\log n}{m}}
        \qquad \text{for all } T\ge 1.
    \end{equation}
\end{theorem}
If we set \(k' = Ck\) in \Cref{thm:tpower-T} for some absolute constant \(C\), the sample complexity requirement \eqref{eq:m-bound-tpower} is weaker than \eqref{eq:m-bound-gamma} in \Cref{thm:main}, so the overall sample complexity is still dominated by \eqref{eq:m-bound-gamma}.

In \Cref{thm:tpower-T}, the number of iterations \(T\) does not appear in the final bound. This means that even a single iteration of TPower refinement suffices to reach the statistical upper bound in terms of order, while further iterations improve the constant factors only and do not improve the rate.

Importantly, the refinement guarantee does not rely on the specifics of SEP. In fact, it is implied from the proof that any initializer \(\hat\bfv\) that aligns with the true spike \(\bfv\) at the constant level \(\sqrt{\gamma}\), i.e., \(|\langle \hat\bfv, \bfv \rangle| \ge \sqrt{\gamma}\), can be upgraded by the TPower refinement with the centered operator \(\hGamma\) to the same statistical error. In this sense, this result is general and can be paired with a variety of polynomial-time initializers. The role of SEP is to furnish such an initializer under broad energy-profile structures.

\subsection{Theoretical superiority of SEP}\label{sec:superiority}
In this section, we discuss the superiority of SEP over existing polynomial-time algorithms in terms of sample complexity across various signal structures. For simplicity, we here ignore the \(\theta\)-dependence and constants, focusing on the order-wise comparison.

A state-of-the-art polynomial-time SPCA algorithm is the single-peak-based method \cite{Cai25PeakPCA} introduced in \Cref{sec:classical}, whose sample complexity is
\begin{equation}\label{eq:ks1-logn}
    m \ \gtrsim \ ks(1) \log n.
\end{equation}

Separately, the related sparse phase retrieval literature derives the profile-dependent scaling~\cite{Xu2025OptimalPR}
\begin{equation}\label{eq:minmax-logn}
    m \ \gtrsim\ \min_{1\le p \le k} \max\left\{p^2s^2(p), k s(p)\right\} \log n.
\end{equation}
Although the bound \eqref{eq:minmax-logn} is not established for SPCA, it is still meaningful to compare it with our result \eqref{eq:m-bound-gamma} since the initialization method in their algorithm shares similar techniques to SPCA algorithms~\cite{Liu2021Towards}. Moreover, the bound \eqref{eq:minmax-logn} is no larger than (and hence no more demanding than) \eqref{eq:ks1-logn}, since \(s(1) \le k\) and further
\begin{align}
     & \min_{1\le p \le k}\max \left\{p^2s^2(p), k s(p)\right\} \nonumber      \\
     & \qquad \le \max \left\{p^2s^2(p), k s(p)\right\} \bigg|_{p=1} = k s(1).
\end{align}

Our next theorem, which may be of independent interest, compares the sample complexity scaling of SEP \eqref{eq:m-bound-gamma}, denoted by \(A(s)\), with the refined reference bound \eqref{eq:minmax-logn}, denoted by \(B(s)\). It establishes that \(A(s)\) uniformly improves upon \(B(s)\), which implies the superiority of SEP over existing polynomial-time algorithms (including the state-of-the-art result in \eqref{eq:ks1-logn}).

\begin{theorem}[Superiority of SEP sample complexity]\label{thm:superiority}
    For any signal-energy structure function \(s(p)\) defined in \Cref{def:structure-function}, define the two quantities
    \begin{align}
        A(s) & := \max_{1\le p \le k} p s^2(p),                             \\
        B(s) & := \min_{1\le p \le k} \max\left\{p^2s^2(p), k s(p)\right\}.
    \end{align}
    Then, the following two statements hold:
    \begin{enumerate}[itemsep=2pt, topsep=2pt, label=(\roman*)]
        \item \textbf{Uniform dominance.} For all profiles \(s(\cdot)\),
              \begin{equation}\label{eq:betterbound}
                  A(s)\ \le\ B(s).
              \end{equation}
        \item \textbf{Strict separation.} There exists a sequence of spikes \(\{\mathbf v^{(k)}\}\) with structure functions \(s_k(\cdot)\) such that
              \begin{equation}\label{eq:strictly-better}
                  \lim_{k\to\infty}\frac{B(s_k)}{A(s_k)}=\infty.
              \end{equation}
    \end{enumerate}
\end{theorem}

%% file: 041Proof.tex
\section{Proof}\label{sec:proof}
We now present the key propositions that form the backbone of our analysis and lead to the proof of \Cref{thm:main}, \Cref{thm:tpower-T}, and \Cref{thm:superiority}. Each proposition serves a distinct role in establishing the sample complexity and refinement guarantees of SEP and TPower. The proofs of these propositions are deferred to \Cref{app:proof-props}.
\subsection{Key Propositions}
We begin with an initialization guarantee that provides the base case for the energy lower bound induction.
\begin{proposition}[Initialization] \label{prop:init-gamma}
    Condition on the high-probability event \(\mcE\). For any \(\gamma\in(0,1)\), if \begin{equation}
        m\ \ge\ C\,\frac{(1+\theta)^2}{\theta^2 (1-\gamma)^2} s^2(1) \log n,
    \end{equation}
    then it holds that
    \begin{equation}
        \|\bfv_{S^{(1)}}\|_2\ \ge\ \sqrt{\frac{\gamma}{s(1)}}.
    \end{equation}
\end{proposition}

Next, we establish the inductive step, showing that the energy lower bound is preserved as the support set is gradually expanded.
\begin{proposition}[Inductive step: energy lower bound preservation]\label{prop:energy-preservation}
    Condition on the high-probability event \(\mcE\). For \(p\in\{1,\dots,k-1\}\) and \(\gamma\in(0,1)\), if
    \begin{equation}
        m \                   \ge\ C\frac{ (1+\theta)^2}{\theta^2 \gamma^2(1-\sqrt{\gamma})^2}(p+1) s^2(p+1) \log n,
    \end{equation}
    and
    \begin{equation}
        \|\bfv_{S^{(p)}}\|_2\ \ge\  \sqrt{\frac{\gamma}{s(p)}},
    \end{equation}
    then the reselected set \(S^{(p+1)}\) (top-\((p{+}1)\) of \(|\hGamma \hat\bfe^{(p)}|\)) satisfies
    \begin{equation}
        \|\bfv_{S^{(p+1)}}\|_2\ge \sqrt{\frac{\gamma}{s(p+1)}}.
    \end{equation}

\end{proposition}

The two propositions above jointly establish the energy lower bound induction described in \Cref{sec:SEP}.
In particular, after \(k\) rounds of support reselection, the final support obeys
\begin{equation}
    \|\bfv_{S^{(k)}}\|_2 \ge \sqrt{\frac{\gamma}{s(k)}} = \sqrt{\gamma}.
\end{equation}
Intuitively, the initialization secures a nontrivial overlap with the true support, and the inductive step guarantees that this overlap cannot deteriorate along the rounds.
This, in turn, underpins the bound on the final direction error in \Cref{thm:main}.

Next, we present the key propositions used in the proof of \Cref{thm:tpower-T}, which analyzes the TPower refinement following the SEP initialization.
The argument proceeds in three stages.
\Cref{prop:alpha0-lb} establishes that the SEP initializer is already well aligned with the true sparse component.
\Cref{prop:HT-align} then quantifies how a single TPower refinement iteration with hard-thresholding affects this alignment.
Finally, \Cref{prop:invariant} combines these results to show that the alignment remains bounded away from zero throughout all iterations, ensuring stable refinement.

We first establish that the SEP initializer achieves a nontrivial correlation with the true sparse component.
\begin{proposition}[Initializer alignment lower bound]\label{prop:alpha0-lb}
    Condition on the high-probability event \(\mcE\). Let \(\bfw^{(0)}\) be the initializer produced in~\Cref{alg:SEP} and \(S^{(k)}\) be the selected support set. Assume that the energy lower bound
    \(\|\bfv_{S^{(k)}}\|_2 \ge \sqrt{\gamma}\) for some \(S^{(k)}\) of size \(k\), with \(\gamma\in(0,1)\).
    It holds that
    \begin{equation}
        \alpha_0:=|\langle \bfw^{(0)},\bfv\rangle|
        \ \ge\
        \sqrt{\gamma}\,\sqrt{\,1-\frac{C_1(1+\theta)^2}{\theta^2\gamma^2}\cdot\frac{k\log n}{m}\,}\, .
    \end{equation}
    In particular, if \(m\ge C\frac{(1+\theta)^2}{\theta^2\gamma^2}k\log n\), then \(\alpha_0 \ge c_0\gamma \) for some absolute \(c_0\in(0,1/2]\).
\end{proposition}

Next, we characterize how a single TPower refinement iteration with hard-thresholding affects the alignment; this will serve as the induction step in our analysis.
\begin{proposition}[Stability and improvement under one hard-thresholding iteration]\label{prop:HT-align}
    Let \(\bfw\in\mathbb{R}^n\) be any \(k'\)-sparse unit vector and set \(\alpha := |\langle \bfw,\bfv\rangle|\).
    Condition on the high-probability event \(\mcE\), and define
    \begin{equation}
        b := C(1+\theta)\sqrt{\frac{k'\log n}{m}}.
    \end{equation}
    Whenever \(\theta\alpha > 2b\), we necessarily have \(\alpha>0\).
    Under the fixed deterministic tie rule, \(\mcH_{k'}(-\bfz)=-\mcH_{k'}(\bfz)\), and since the direction metrics are invariant under a global sign, replacing \(\bfw\) by \(-\bfw\) if necessary allows us to assume without loss of generality that \(\langle\bfw,\bfv\rangle=\alpha\).
    Consider
    \begin{equation}
        \bfy = \hGamma\bfw = \theta\alpha\bfv + \bfxi,
        \quad \text{where } \bfxi := \bfW\bfw.
    \end{equation}
    Then the normalized hard-thresholded vector satisfies
    \begin{equation}\label{eq:cos-lb}
        \cos\angle\!\left(\frac{\mathcal{H}_{k'}(\bfy)}{\|\mathcal{H}_{k'}(\bfy)\|_2},\,\bfv\right)
        \ \ge\ \frac{\theta\alpha - 2b}{\theta\alpha + b},
    \end{equation}
    and
    \begin{equation}\label{eq:sin-up}
        \sin\angle\!\left(\frac{\mathcal{H}_{k'}(\bfy)}{\|\mathcal{H}_{k'}(\bfy)\|_2},\,\bfv\right)
        \ \le\ \frac{5b}{\theta\alpha - 2b}.
    \end{equation}
    Consequently, when the signal strength \(\theta\alpha\) dominates the noise level \(b\)
    (e.g., \(m\gtrsim\frac{(1+\theta)^2}{\theta^2\gamma^2}k\log n\) makes \(b\) sufficiently small), one hard-thresholding iteration preserves the alignment with \(\bfv\).
\end{proposition}

Combining the initialization guarantee and the one-iteration refinement bound, we obtain an invariant that ensures the alignment remains bounded across all iterations.
\begin{proposition}[Alignment invariant across iterations]\label{prop:invariant}
    Let \(\bfw^{(t)}\) be the \(t\)-th iterate of~\Cref{alg:tpower-refine} and \(\alpha_t:=|\langle \bfw^{(t)},\bfv\rangle|\). When \(m\ge C\frac{(1+\theta)^2}{\theta^2\gamma^2}k\log n\) with \(C\) sufficiently large, there exists an absolute \(c_\ast\in(0,1/2]\) such that for all \(t\ge 0\),
    \begin{equation}
        \alpha_t \ \ge\ c_\ast \gamma.
    \end{equation}
\end{proposition}

Now we are ready to prove our main theorems.

\subsection{Proof of~\Cref{thm:main}}
\begin{proof}
    Condition on the high-probability event of~\Cref{prop:principal-submatrix}. By~\Cref{prop:init-gamma}, the initialization selects an index \(S^{(1)}\) such that \(\|\bfv_{S^{(1)}}\|_2\ge \sqrt{\gamma/s(1)}\), provided
    \begin{equation}\label{eq:init-bound}
        m\ \ge\ C\,\frac{(1+\theta)^2}{\theta^2 (1-\gamma)^2} s^2(1) \log n.
    \end{equation}
    Now assume for some \(p\in\{1,\dots,k-1\}\) that \(\|\bfv_{S^{(p)}}\|_2\ge \sqrt{\gamma/s(p)}\). Applying \Cref{prop:energy-preservation}, we see that the reselection preserves the energy lower bound at level \(p{+}1\) whenever
    \begin{equation}\label{eq:induction-bound}
        m\ \ge \ C \frac{ (1+\theta)^2}{\theta^2 \gamma^2(1-\sqrt{\gamma})^2}(p+1) s^2(p+1)\log n.
    \end{equation}
    Imposing the uniform bound \eqref{eq:induction-bound} ensures that this condition holds for every \(p\le k{-}1\), hence by induction we obtain \(\|\bfv_{S^{(k)}}\|_2\ge \sqrt{\gamma/s(k)}=\sqrt{\gamma}\). Since \(\gamma\in(0,1)\), combining \eqref{eq:init-bound} and \eqref{eq:induction-bound} yields the uniform sample size requirement \eqref{eq:m-bound-gamma}.

    For the final direction error, by the triangle inequality for principal angles,
    \begin{equation}
        \sin\angle(\hat \bfv,\bfv)\ \le\ \sin\angle(\hat \bfv,\bfu_S)\ +\ \sin\angle(\bfu_S,\bfv),
    \end{equation}
    where \(\bfu_S:=\bfv_{S^{(k)}}/\|\bfv_{S^{(k)}}\|_2\). The second term equals \(\|\bfv_{S^{(k)}{}^c}\|_2\le \sqrt{1-\gamma}\), and applying~\Cref{lem:DK-sinTheta} on the first term gives \(\sin\angle(\hat \bfv,\bfu_S)\le \|\bfW_{S^{(k)},S^{(k)}}\|_2/(\theta\,\|\bfv_{S^{(k)}}\|_2^2)\). Hence
    \begin{align}
        \sin\angle(\hat \bfv,\bfv)
         & \le \sqrt{1-\gamma}
        + \frac{\|\bfW_{S^{(k)},S^{(k)}}\|_{2}}
        {\theta\,\|\bfv_{S^{(k)}}\|_2^2} \nonumber \\
         & \le \sqrt{1-\gamma}
        + \frac{C(1+\theta)}{\theta\,\gamma}
        \sqrt{\frac{k\,\log n}{m}}.
    \end{align}
    This matches the stated bound and in particular vanishes as \(m\to\infty\) under \eqref{eq:m-bound-gamma}.
\end{proof}

\subsection{Proof of~\Cref{thm:tpower-T}}

\begin{proof}
    We prove by induction on \(t\) that
    \begin{equation}\label{eq:tpower-induction-claim}
        \anglesin \big(\bfw^{(t+1)},\bfv\big)\ \le\ C\,\frac{1+\theta}{\theta\gamma}\,\sqrt{\frac{k'\log n}{m}}
        \qquad\text{for all } t\ge 0.
    \end{equation}
    Fix \(t\). Apply~\Cref{prop:HT-align} to \(\bfw^{(t)}\):
    \begin{equation}
        \anglesin\left(\frac{\mcH_{k'}(\bfy)}{\|\mcH_{k'}(\bfy)\|_2},\,\bfv\right)
        \ \le\ \frac{5b}{\theta\alpha-2b}.
    \end{equation}
    By~\Cref{prop:invariant}, \(\alpha\ge c_\ast\gamma\) for all \(t\). Hence, when \(m\ge C\frac{(1+\theta)^2}{\theta^2\gamma^2}k\log n\) with \(C\) sufficiently large, we have
    \begin{equation}
        \anglesin\big(\bfw^{(t+1)},\bfv\big)
        \ \le\ C\,\frac{1+\theta}{\theta \gamma}\sqrt{\frac{k'\log n}{m}}.
    \end{equation}
    Since the bound~\eqref{eq:tpower-induction-claim} is independent of \(t\), it holds in particular at \(t=T-1\), which yields the theorem.
\end{proof}

\subsection{Proof of~\Cref{thm:superiority}}
\begin{proof}
    First, we show the uniform dominance \eqref{eq:betterbound}.
    Let \(q^\star \in\{1,2,\dots,k\}\) be an index such that \(A(s)\) is maximized at \(p=q^\star\), so that \(A(s) = q^\star s^2(q^\star)\). For any \(p\in\{1,2,\dots,k\}\), we first show that
    \begin{equation}\label{eq:qpbound}
        q^\star s^2(q^\star)\ \le\ \max \{p^2s^2(p), ks(p)\}.
    \end{equation}
    \begin{enumerate}
        \item Consider the case where \(q^\star \le p\). Since \(ps(p)\) is non-decreasing in \(p\), we have
              \begin{equation}
                  q^\star s^2(q^\star) = (q^\star)^2 s^2(q^\star) / q^\star  \le p^2 s^2(p) / q^\star \le p^2 s^2(p).
              \end{equation}
        \item Consider the case where \(q^\star > p\). Since \(s(p)\) is non-increasing in \(p\) and \(q^\star s(q^\star) \leq k\), we have
              \begin{equation}
                  q^\star s^2(q^\star) = q^\star s(q^\star) \cdot s(q^\star) \le k s(p).
              \end{equation}
    \end{enumerate}
    Combining the two cases above, we obtain \eqref{eq:qpbound}. Taking minimum over \(p\in\{1,2,\dots,k\}\) on the right-hand side of \eqref{eq:qpbound} yields
    \begin{equation}
        A(s) = q^\star s^2(q^\star) \le \min_{1\le p \le k} \max\{p^2s^2(p), ks(p)\} = B(s),
    \end{equation}
    which proves \eqref{eq:betterbound}.

    Next, we show the strict separation \eqref{eq:strictly-better}. We construct a sequence of power-law decaying signals. Let
    \begin{equation}
        v_{(j)}^2 = \left(\sum_{j=1}^k j^{-1/2}\right)^{-1}j^{-1/2}, \quad j=1,2,\dots,k
    \end{equation}
    be the non-zero entries of \(\bfv^{(k)}\). From the proof of \Cref{prop:power-law-interp} (see \eqref{eq:power-law-sp}), the structure function \(s_k(p)\) of \(\bfv^{(k)}\) satisfies
    \begin{equation}
        s_k(p) \asymp \sqrt{\frac{k}{p}}.
    \end{equation}

    Now we bound \(A(s_k)\) and \(B(s_k)\). For \(A(s_k)\), we have
    \begin{equation}
        A(s_k) = \max_{1\le p \le k} p s_k^2(p) \asymp k.
    \end{equation}
    For \(B(s_k)\), we have
    \begin{align}
        B(s_k) & = \min_{1\le p \le k} \max\{p^2 s_k^2(p), ks_k(p)\} \nonumber \\
               & \asymp  \min_{1\le p \le k} \max\{pk, k^{3/2}p^{-1/2}\}
    \end{align}
    The minimum is attained at \(p \asymp k^{1/3}\), which gives \(B(s_k) \asymp k^{4/3}\).
    Therefore, we have constructed a signal such that \(A(s_k) \asymp k\) and \(B(s_k) \asymp k^{4/3}\), which completes the proof.
\end{proof}

%% file: 05Discussion.tex

\section{Discussion}\label{sec:discussion}

\subsection{Why SEP is better: selection rule and \(\ell_2\) perturbation control}\label{sec:dis-view-comparision}
In this section, we explain why SEP achieves the lower order \(\max_{p\le k} p\,s^2(p) \log n\) while diagonal screening typically yields \(k^2\log n\). Two components jointly determine the rate.

\begin{enumerate}[itemsep=2pt, topsep=2pt, label=(\roman*)]
    \item \emph{Selection rule.} SEP selects the support using the full response
          \begin{equation}
              \hGamma \hat \bfe^{(p)} \;=\; \theta\,\langle \bfv,\hat \bfe^{(p)}\rangle\, \bfv \;+\; \bfW \hat \bfe^{(p)},
          \end{equation}
          which aggregates information across coordinates and is naturally compared in the sense of \(\ell_2\) norm.
          Diagonal screening scores coordinates by the diagonals \(\{\hGamma_{jj}\}\) and thus enforces per-coordinate separation.

    \item \emph{\(\ell_2\) perturbation control.} We establish a uniform operator-norm bound on all principal blocks of \(\bfW\) (see \Cref{prop:principal-submatrix}):
          \begin{equation}
              \|\bfW_{S,S}\|_2 \;\lesssim\; \sqrt{|S|\log n/m}\qquad \forall S\subset[n].
          \end{equation}
          Due to the natural relationship between the \(\ell_2\) norm and operator norm, this uniform bound yields a clean \(\ell_2\) bound on the noise term \(\bfW \hat \bfe^{(p)}\), thus allowing us to analyze the energy on the reselected support and the pathwise progress of SEP across rounds.
          By contrast, if one instead employs an entrywise control on the noise term \(\bfW \hat \bfe^{(p)}\), one needs to establish a uniform \(\ell_\infty\) bound over all data-driven selections and all \(p\)-sparse supports so that the coordinate-wise margin can hold simultaneously for \(p\) coordinates at each round and across rounds (as required under data dependence, see \Cref{sec:data-dependence}). This uniformity typically inflates the requirement by about a factor \(p\) (up to logarithmic terms), effectively turning \(ps^2(p)\) into \(p^2 s^2(p)\).

          For diagonal screening, one needs to compare each strong coordinate to the per-entry noise scale \(\sqrt{\log n/m}\). Given only the total energy \(\sum_{j\le p} v_{(j)}^2\asymp 1/s(p)\), the most favorable allocation assumption across the top \(p\) gives \(v_{(p)}^2 \asymp 1/(p\,s(p))\), which leads to
          \begin{equation}
              m \;\gtrsim\; p^2 s^2(p)\,\log n.
          \end{equation}
          Interestingly, an energy-based variant that tracks the diagonal sum \(\sum_{j\in S}\hGamma_{jj}\) can avoid the assumption where \(v_{(p)}^2 \asymp 1/(ps(p))\). For any fixed subset \(S\) of size \(p\), concentration of sums yields fluctuations of order \(\sqrt{p\log n/m}\), which would suggest a \(ps^2(p)\) scaling for a fixed \(S\). Yet the data-driven choice is the maximizer over \(\binom{n}{p}\) subsets; a uniform bound incurs an additional term \(\log\binom{n}{p}\asymp p\log(n/p)\) and inflates the deviation by an additional \(\sqrt{p}\) (equivalently, multiplies the required \(m\) by \(p\)). Thus the overall rate remains \(p^2 s^2(p)\log n\) up to constants.

          Finally, in practice one often sets \(p=k\) for stable estimation, which gives the classical \(k^2 s^2(k)\log n=k^2\log n\) rate.
\end{enumerate}

In conclusion, the improvement of SEP comes from the combination of a response-based selection rule and an energy-based analysis that controls \(\bfW\) in operator norm.

\subsection{On the role of TPower refinement and the choice of operator}\label{sec:dis-tpower}
A key consequence of our analysis is that the final statistical error after TPower refinement is independent of the number of refinement iterations \(T\); see \Cref{thm:tpower-T}. In fact, a single iteration already reaches the statistical upper bound. This one-iteration phenomenon relies on two ingredients we already established in the main proof: (i) a sharp bound in~\Cref{prop:HT-align} that controls the numerator/denominator after the spectral update and keep-\(k'\) hard-thresholding that preserves a constant alignment; and (ii) using the centered operator \(\hGamma=\hSigma-\bfI=\theta\, \bfv\bfv^\top+\bfW\), so that the spectral update contains no carry-over term:
\begin{equation}
    \bfy\ = \ \hGamma\bfw\ =\underbrace{\theta\,\langle \bfw,\bfv\rangle\,\bfv}_{\text{signal aligned with } \bfv}+\ \underbrace{\bfW\bfw}_{\text{noise}}.
\end{equation}
Given that the initializer has constant alignment \(\alpha_0\gtrsim\sqrt{\gamma}\) (\Cref{prop:alpha0-lb}), the signal component already points exactly along \(\bfv\) with strength \(\theta\alpha_0\), and hard-thresholding (\Cref{prop:HT-align}) attenuates the noise on the selected \(k'\) coordinates down to the order of \((1+\theta)\sqrt{k'\log n/m}\). This yields immediately the clean statistical error in \Cref{thm:tpower-T}, hence \(T\) disappears from the bound.

By contrast, if either (i) one uses coarser, black-box contraction analyses (e.g., as in~\cite{Yuan2013TPower, Cai25PeakPCA}), or (ii) one refines with the raw covariance \(\hSigma\) instead of \(\hGamma\), the spectral update contains an additional carry-over term since there is an identity component \(\bfI\) in \(\hSigma\):
\begin{equation}
    \bfy'\ =\ \hSigma \bfw\ =\ \underbrace{\bfw}_{\text{carry-over}}+\ \underbrace{\theta\, \langle \bfw,\bfv\rangle \bfv}_{\text{signal}}\ +\ \text{noise}.
\end{equation}
This carry-over persists after thresholding and contaminates both numerator and denominator in the alignment ratio, producing a genuine per-iteration residual that must be iteratively damped. At the level of orders, this leads to a recursion of the familiar form
\begin{align}
    \sin\angle(\bfw^{(t+1)},\bfv)
     & \ \le\ \underbrace{\rho \sin\angle(\bfw^{(t)},\bfv)}_{\text{optimization error}} \nonumber                  \\
     & \qquad +\ \underbrace{C \frac{1+\theta}{\theta \gamma}\sqrt{\frac{k'\log n}{m}}}_{\text{statistical error}}
\end{align}
with \(\rho < 1\). Thus, an explicit optimization term remains until \(t\) is large. This is the standard behavior in the existing literature on iterative SPCA methods. Our analysis shows that this can be avoided by using the centered operator \(\hGamma\) and a careful, iteration-invariant perturbation analysis.

\subsection{Data dependence}\label{sec:data-dependence}
Every round of \Cref{alg:SEP} is data-dependent: the support \(S^{(p)}\) at round \(p\) is selected from the intermediate response \(\bfu^{(p-1)}=\hGamma \hat \bfe^{(p-1)}\), which itself is computed from the same sample. As already visible in \eqref{eq:peakiness-proxy} for the single-peak heuristic, even the choice of \(j_{\max}\) is a function of the data. Consequently, one cannot treat the iterates as independent of the sample when taking expectations or applying tail bounds directly.

Our analysis handles this dependence uniformly along the entire path. Specifically, for the response term in \eqref{eq:SEP-response} we establish a single high-probability event \(\mcE\) under which the operator norm of every relevant principal noise block is controlled:
\begin{align*}
    \|\bfW_{S,S}\|_2 & \;\lesssim\; \sqrt{|S|\log n/m}, \nonumber  \\
                     & \text{simultaneously for all } S\subset[n].
\end{align*}
(see \Cref{prop:principal-submatrix}). Hence \(\mcE\) holds for all rounds \(p\le k\) and all data-driven supports \(S^{(p)}\). This uniform spectral control allows us to reason about the captured energy on the selected block without conditioning on the data.

Alternative techniques, such as leave-one-out arguments~\cite{Ma2020LOO} or perturbative decouplings~\cite{arcones1993decoupling}, could also address the data dependence, but they are considerably more involved. Our approach remains concise and directly tied to the selection mechanism.

Finally, we remark that it is sometimes convenient to proceed as if a small number of iterates were independent of the data, especially when only a constant number of data-dependent choices are made. This can be repaired by a fixed number of sample splits. Specifically, using a fresh block whenever a data-dependent choice occurs can address the data dependence issue and preserve the asymptotic order (see, e.g.,~\cite{DiCiccio2020DataSplit}). In our setting, however, the selection is repeated over \(k\) rounds; an analogous repair would require \(k\) disjoint splits, increasing the sample complexity by a factor of \(k\).

\subsection{Connection to weak sparsity and \(\ell_q\) balls}
The structure function \(s(p)\) complements classical weak-sparsity descriptions based on \(\ell_q\) balls, \(0<q<1\)~\cite{VuLei2012MinimaxSPCA}. An \(\ell_q\)-ball constraint gives a class-level, global envelope for ordered coordinates, providing a compact summary of coefficient decay. By definition, \(s(p)^{-1}=\sum_{i=1}^p v_{(i)}^2\) records cumulative top-\(p\) energy at every truncation level. The two descriptions are complementary: the former summarizes a signal class, whereas the latter retains profile-level information aligned with SEP's pathwise support growth across intermediate sizes.

\subsection{Limitations and open directions}
We scope this work to the single-spike model and develop SEP with structure-adaptive guarantees (\Cref{thm:main}). Several questions remain open.
\begin{enumerate}[itemsep=2pt, topsep=2pt, label=(\roman*)]
    \item \textbf{Information-theoretic lower bounds.} While we establish a better polynomial-time upper bound tied to the energy profile via \(s(p)\), matching information-theoretic lower bounds under general profiles remain open.
    \item \textbf{Statistical-computational tradeoffs beyond flat spikes.} Although the classical gap is well understood in the flat regime, a systematic characterization under non-flat profiles remains to be developed. This includes whether stronger concentration collapses, narrows, or reshapes the gap, and whether new barriers arise.
    \item \textbf{Dependence on signal strength \(\theta\).} Although recovery should intuitively become increasingly easy as \(\theta\to\infty\), our guarantees do not reflect this behavior: both the sample-complexity prefactor \((1+\theta)^2/\theta^2\) and the \(\theta\)-dependent error prefactor \((1+\theta)/\theta\) saturate. This limitation arises partly because our Davis--Kahan step controls the full spectral norm of the covariance noise. Suppressing support-size and logarithmic factors, this noise contains a signal-parallel fluctuation of order \(\theta/\sqrt{m}\), which mainly changes the leading eigenvalue, and signal-orthogonal cross-terms of order \(\sqrt{\theta/m}\), which drive eigenvector rotation. Bounding both by the former yields a saturated angular estimate, whereas isolating the cross-terms suggests an error improving as \(1/\sqrt{\theta m}\). Strong-signal asymptotics are also outside the scope of related analyses, which either keep the spike eigenvalues fixed~\cite{Birnbaum2013MinimaxSPCA} or derive minimax detection lower bounds in a weak-signal regime~\cite{BerthetRigollet2013aSPCADetection}. Obtaining a sharper bound uniformly over SEP's data-dependent supports requires a finer perturbation analysis and is left open.
    \item \textbf{Complexity-theoretic routes.}
          Two complementary avenues are promising. First, one can extend existing planted-clique reductions~\cite{BerthetRigollet2013bSPCADetectionLB,Wang2016SPCACompStatTradeoff}, which underlie the flat-case gap, to models that encode signal structure (e.g., nonuniform or weighted supports), thereby obtaining hardness for SPCA under non-flat structures. Second, one can directly study weighted planted-clique (or planted-subgraph) models whose weights reflect the energy profile. Hardness or tractability results there would transfer to structure-adaptive SPCA. Conversely, progress on structure-adaptive SPCA (e.g., sharp statistical limits and algorithms matched to energy profiles) may inform the design and analysis of weighted planted problems, suggesting a two-way connection between structure-aware estimation and planted-subgraph complexity. Beyond the single-spike setting, extending the analysis to general background covariance matrices and to multi-spike subspaces is a natural direction. Our discussion of the centered operator for TPower suggests that parts of the analysis may carry over under appropriate spectral corrections, while a complete treatment is left for future work. Finally, we expect the present analysis to shed light on other structure-adaptive estimation problems with similar problem formulations, such as sparse phase retrieval~\cite{Jagatap2019Sample,Cai2022Provable,Liu2021Towards} and sparse canonical correlation analysis~\cite{Hardoon2011SCCA}.
\end{enumerate}

%% file: 06Simulations.tex
\section{Simulations}\label{sec:simulations}
\subsection{Setup}
We evaluate SEP against two strong baselines under the standard single-spike model.
We consider three signal profiles on the true support \(S^\star\):
\begin{itemize}
    \item {\bf Flat}: \(v_j=k^{-1/2}\) on \(S^\star\);
    \item {\bf (Offset) power-law}: \(v_j \propto j^{-1/2}+0.1\) on \(S^\star\), then normalized;
    \item {\bf (Offset) exponential}: \(v_j \propto e^{-j}+0.1\) on \(S^\star\), then normalized;
\end{itemize}
where indices \(j\) are ordered by magnitude within \(S^\star\). The offset 0.1 avoids vanishingly small entries and yields a decaying head followed by a near-flat floor for the non-flat cases. The resulting \(s(p)\) curves qualitatively reflect the profile-dependent behavior highlighted in \Cref{prop:power-law-interp}. \Cref{fig:s_profiles} compares \(s(p)\) for the three profiles with the sparsity \(k=40\).

\begin{figure}[t]
    \centering
    \includegraphics[width=0.7\linewidth]{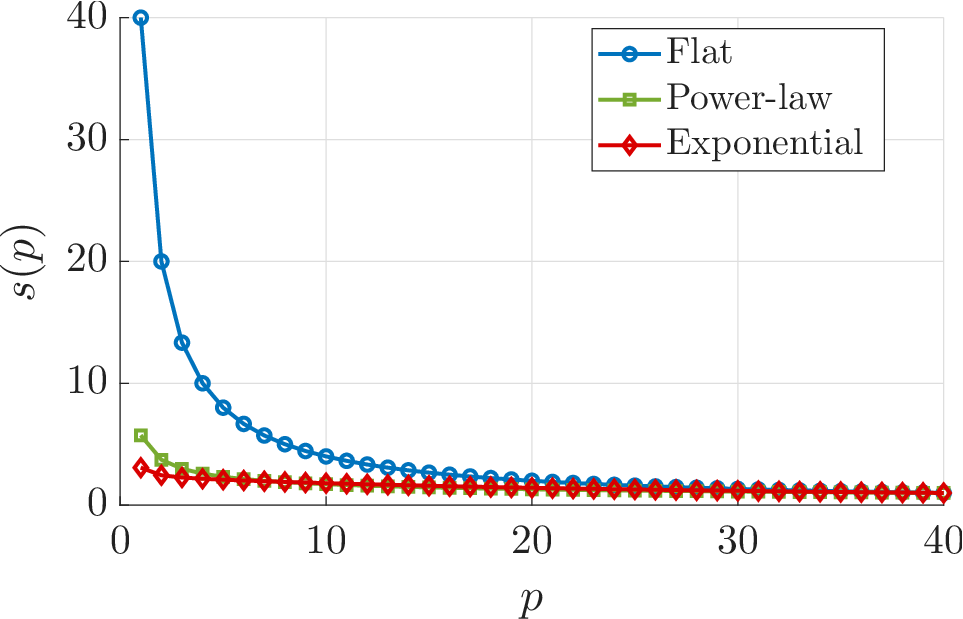}
    \caption{\(s(p)\) for \(k=40\) under the three profiles. The power-law and exponential curves start lower at small \(p\), reflecting stronger early concentration.}
    \label{fig:s_profiles}
\end{figure}

We vary \(m\) from \(100\) to \(1000\) with step \(50\) and \((n,k,\theta)=(1000,40,3)\). The TPower refinement (\Cref{alg:tpower-refine}) uses \(10\) iterations and all three algorithms employ the centered covariance matrix \(\hGamma\). We choose DT and single-peak-based algorithms as competitive baselines. Two metrics are used: direction error \(\sin\angle(\hat\bfv,\bfv)\) and support recall \(|S\cap S^\star|/k\). For each \(m\) and signal profile, we repeat the experiment for \(1000\) trials to average out randomness.

\subsection{Performance across profiles}

\begin{figure*}[t]
    \centering
    \subfloat[Flat signals]{\includegraphics[width=0.3\linewidth]{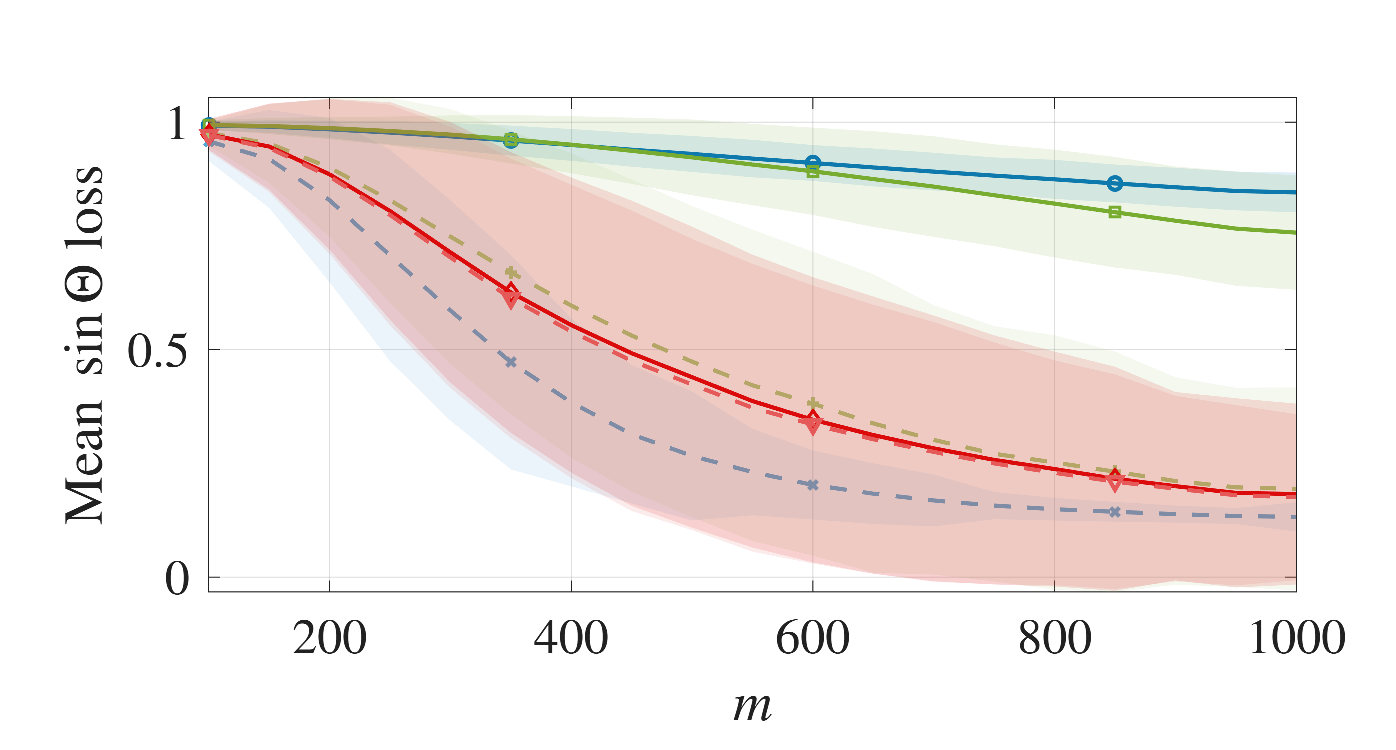}\label{fig:vary_m_err_flat}}
    \hfill
    \subfloat[Power-law decaying signals]{\includegraphics[width=0.3\linewidth]{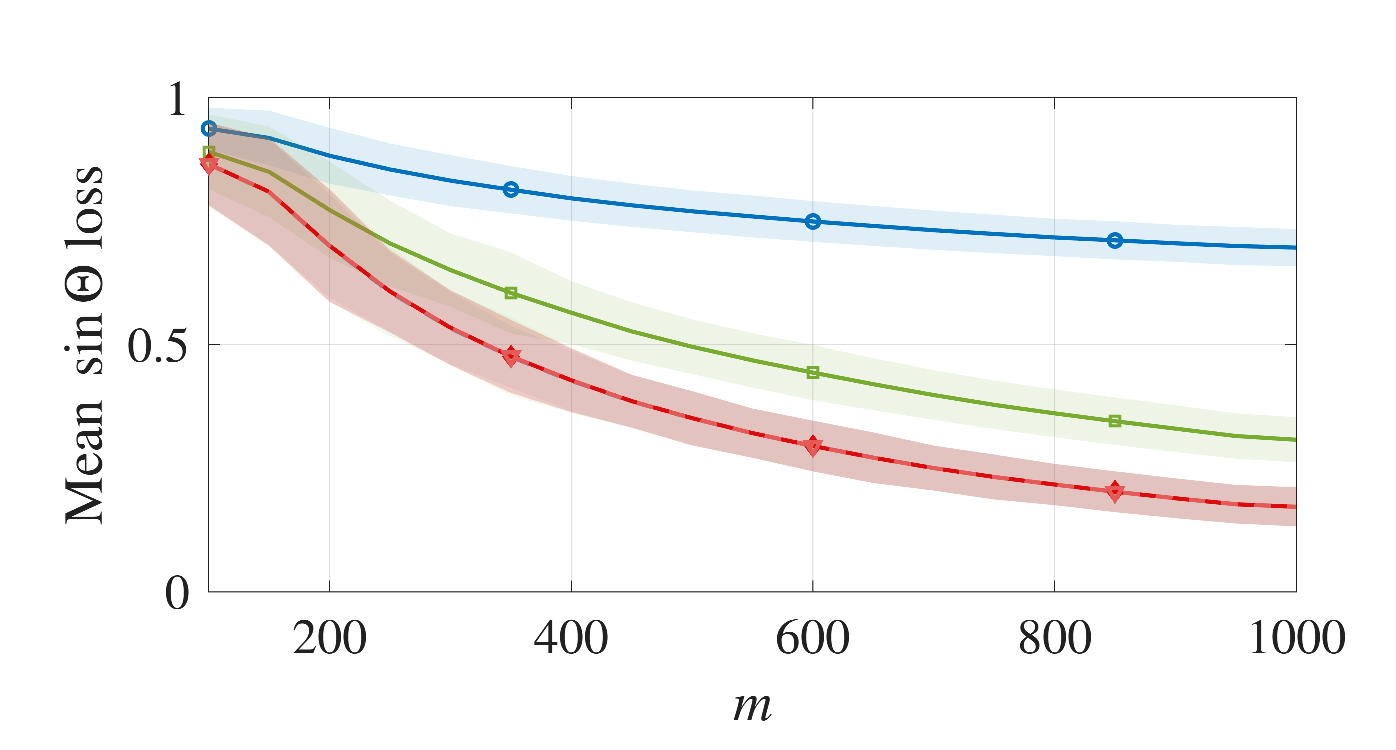}\label{fig:vary_m_err_jhalf}}
    \hfill
    \subfloat[Exponential decaying signals]{\includegraphics[width=0.3\linewidth]{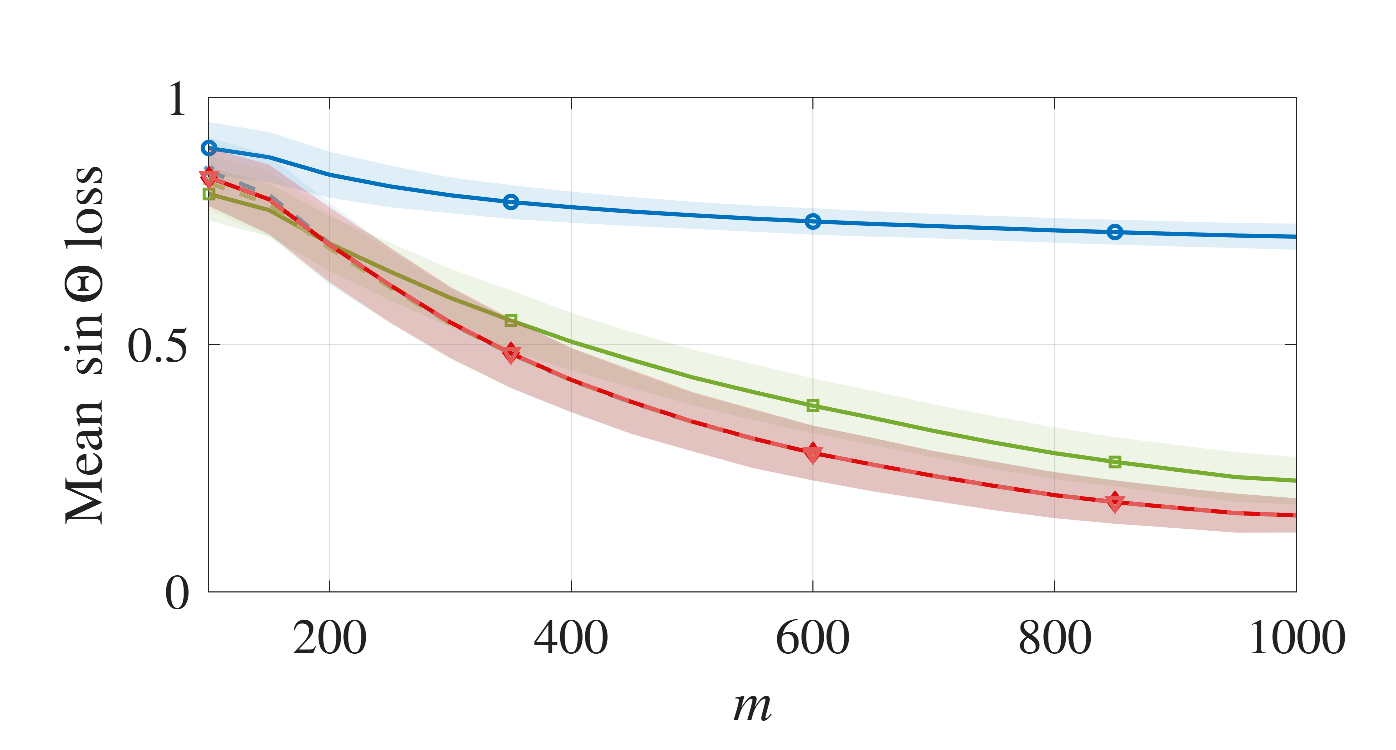}\label{fig:vary_m_err_exp}}
    \hfill
    \raisebox{7mm}{\includegraphics[width=0.08\linewidth]{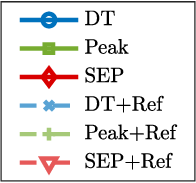}}
    \caption{Direction error vs \(m\) across three profiles (curves: trial mean; shaded bands: \(\pm1\) standard deviation over \(1000\) trials).}
    \label{fig:vary_m_err_row}
\end{figure*}

\begin{figure*}[t]
    \centering
    \subfloat[Flat signals]{\includegraphics[width=0.3\linewidth]{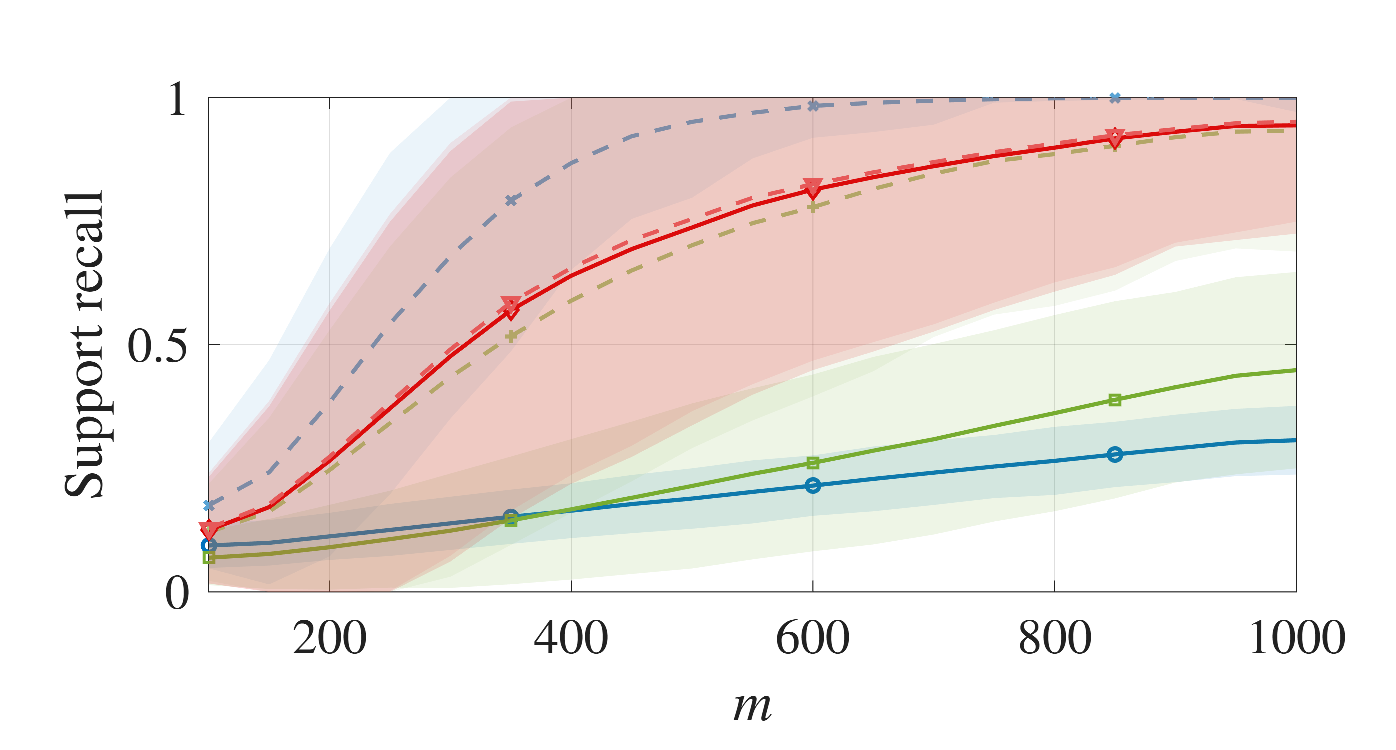}\label{fig:vary_m_hit_flat}}
    \hfill
    \subfloat[Power-law decaying signals]{\includegraphics[width=0.3\linewidth]{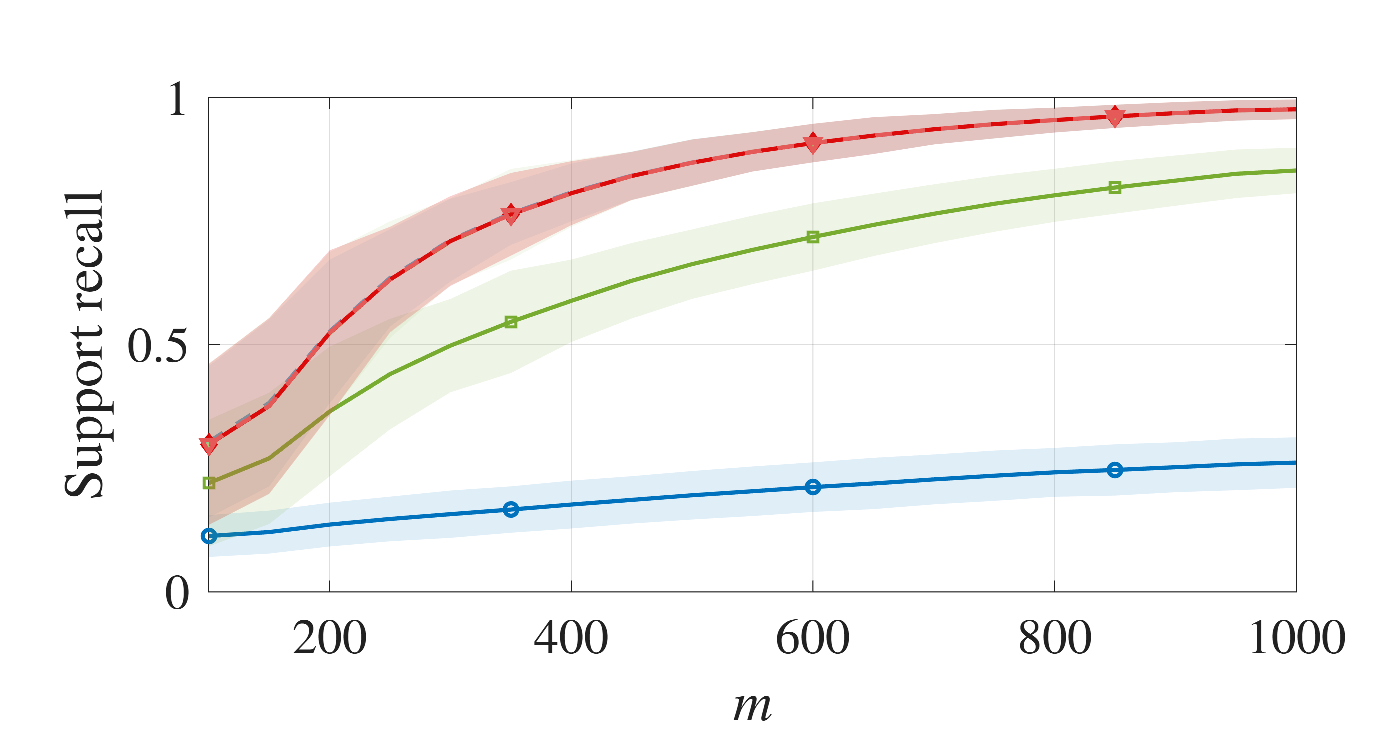}\label{fig:vary_m_hit_jhalf}}
    \hfill
    \subfloat[Exponential decaying signals]{\includegraphics[width=0.3\linewidth]{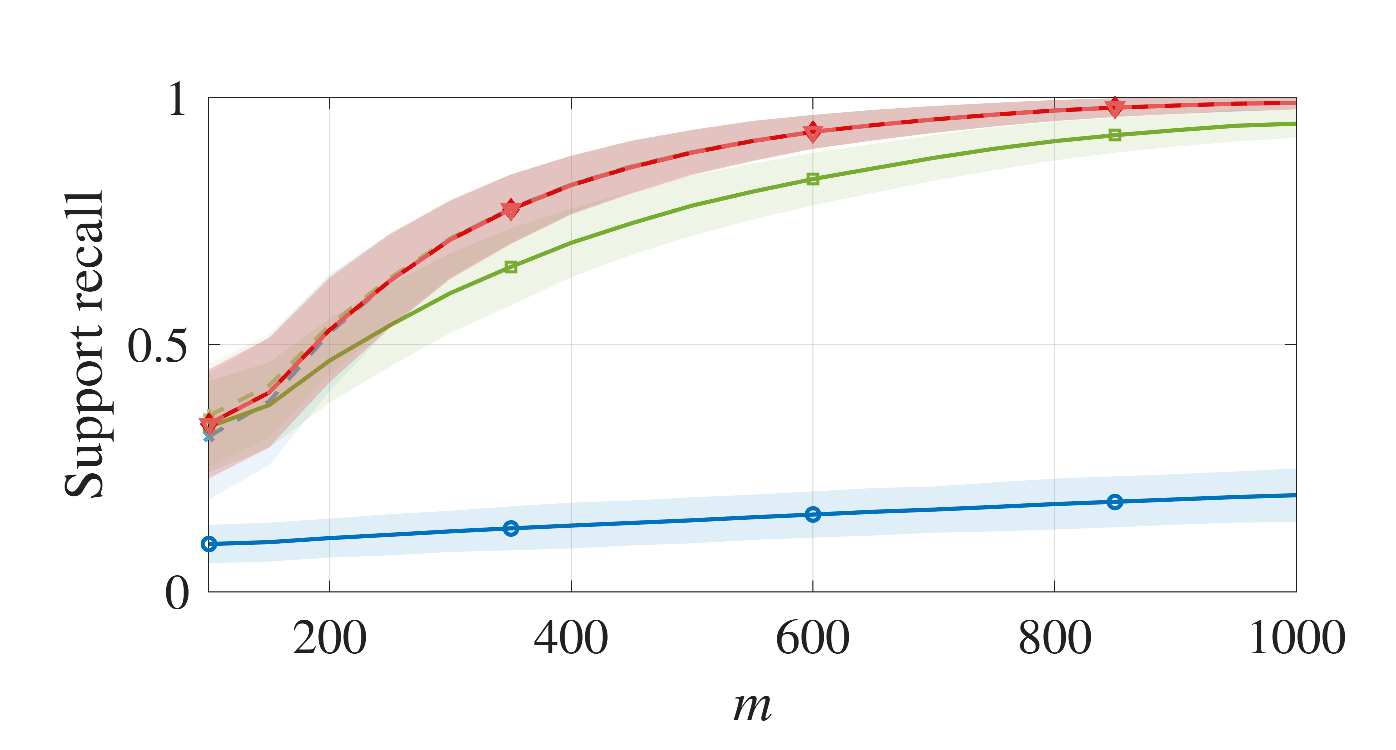}\label{fig:vary_m_hit_exp}}
    \hfill
    \raisebox{7mm}{\includegraphics[width=0.08\linewidth]{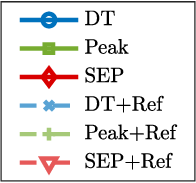}}
    \caption{Support recall vs \(m\) across three profiles (curves: trial mean; shaded bands: \(\pm1\) standard deviation over \(1000\) trials).}
    \label{fig:vary_m_hit_row}
\end{figure*}

\Cref{fig:vary_m_err_row} and \Cref{fig:vary_m_hit_row} show consistent trends across profiles and refinement. For the flat profile, SEP yields the lowest pre-refinement error whereas DT is marginally better after TPower refinement. This suggests that while SEP provides a superior initial subspace estimate, its aggressive support selection might occasionally lock into a suboptimal support set that TPower cannot fully correct, whereas DT maintains a broader initial uncertainty that TPower can leverage. Nevertheless, both methods achieve comparable final accuracy. In contrast, for the power-law and exponential profiles, SEP is decisively best before refinement. After refinement, the curves nearly coincide: SEP changes little, whereas DT and the single-peak method improve to the same level. Overall, refinement primarily benefits weaker initializers, while a strong initializer (SEP) is already near its statistical limit and thus insensitive to additional iterations in our setting.

Excluding post-refinement, SEP achieves the best performance across all profiles. Moreover, comparing the two non-flat profiles, SEP's margin over the single-peak method is larger for power-law signals and noticeably smaller for exponential signals. The reason is structural: the exponential profile has a more pronounced leading coordinate, which the single-peak heuristic captures well; power-law decay distributes energy across the leading coordinates, so restricting attention to the largest entry fails to exploit the information carried by the other prominent coordinates, whereas SEP adapts to the entire profile. This observation is consistent with the theoretical prediction in \Cref{thm:superiority}, suggesting that SEP's advantage is most pronounced when the signal structure is non-flat yet not trivially single-peaked.

\subsection{TPower refinement and centering}\label{sec:refineDT}
\Cref{fig:vary_m_err_row} shows that with SEP as the initializer, adding \(T=10\) TPower iterations (\Cref{alg:tpower-refine}) brings virtually no change---the initializer is already near its statistical upper bound.
To reveal the effect of TPower and the centering covariance, in this subsection we switch the initializer to DT and study refinement from a weaker \(\hat\bfv\). We run TPower with two matrices:
\begin{equation}
    \text{Uncentered: } \hSigma,\qquad
    \text{Centered: } \hGamma=\hSigma-\bfI.
\end{equation}
We set \((n,k,m,\theta)=(1000,40,400,3)\) and consider flat signals. We set the maximum number of refinement iterations to \(T=100\) and report the mean \(\sin\angle(\bfw^{(t)},\bfv)\) over \(1000\) trials versus \(t\); see \Cref{fig:refine_center_vs_sigma_DT}.

\begin{figure}[t]
    \centering
    \includegraphics[width=0.78\linewidth]{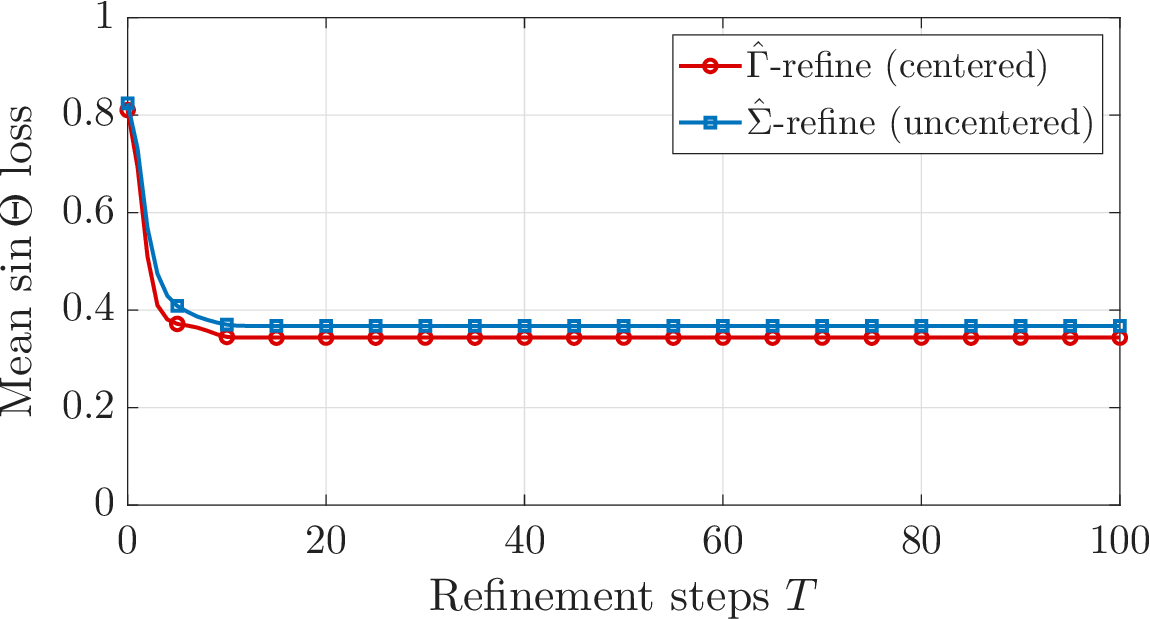}
    \caption{Refinement from a DT initializer: centered \(\hGamma\) vs.\ uncentered \(\hSigma\).
        Both variants decrease within a few iterations. The centered operator attains a lower error floor.}
    \label{fig:refine_center_vs_sigma_DT}
\end{figure}

It can be observed that both variants contract rapidly during the first few iterations and then plateau, and the centered operator consistently achieves a lower error floor. This reflects the superiority of the centered covariance in the refinement, as also discussed in \Cref{sec:dis-tpower}. However, the decrease is not strictly ``one iteration''---small but visible gains appear from \(T=1\) to \(T=10\), which seems to contradict \Cref{thm:tpower-T}. To explain this, we note two factors regarding the theoretical versus numerical behavior. First, \Cref{thm:tpower-T} provides an order-wise guarantee. While a single centered iteration suffices to attain the optimal error rate (i.e., the order \(\sqrt{k/m}\)), subsequent iterations continue to optimize the leading constant factors, numerically reducing the error until a fixed point is reached. Second, the theorem relies on a thresholded entrance condition. It states that once the iterate satisfies an energy threshold, a single step reaches the statistical error:
\begin{align}
    |\langle \bfw^{(0)},\bfv\rangle| \ge \sqrt{\gamma}
     & \quad\Longrightarrow \nonumber                                                   \\
    \sin\angle(\bfw^{(t)},\bfv)
     & \ \lesssim\ \text{statistical error}, \text{for all } t\ge 1,\label{eq:entrance}
\end{align}
where \(\gamma\) is the constant in \Cref{thm:tpower-T}. With DT as the initializer, this entrance condition is typically not met at $t=0$; the first few iterations act as a ``warm-up'' to reselect support and enter the basin of attraction. By contrast, SEP usually satisfies~\eqref{eq:entrance} at $t=0$ (see~\Cref{fig:vary_m_err_row}), so additional refinement yields negligible gains, consistent with the theorem.

%% file: 07Conclusion.tex
\section{Conclusion}\label{sec:conclusion}
We introduced SEP, a simple iterative algorithm for SPCA. Despite requiring no profile information at run time, our analysis formalizes the role of the signal's energy profile via the structure function \(s(p)\) and establishes guarantees for direction estimation with sample size scaling as \(\max_{1\le p\le k} ps^2(p)\log n\). This requirement is uniformly no larger than prior polynomial-time bounds and can be strictly smaller on broad non-flat families. With a single TPower refinement, the final estimator attains the statistical error order. Empirically, SEP outperforms diagonal thresholding and single-peak-based methods, with especially strong gains on non-flat profiles.

%% file: 08Appendix.tex
\appendices\label{sec:appendix}
\crefalias{section}{appendix}
\crefalias{subsection}{appendix}

\section{Auxiliary Lemmas}\label{app:lemmas}
We first give a lemma to establish the asymptotic equivalence among several conditions involving the structure function \(s(p)\), which will be useful to translate different forms of sample complexity conditions in the analysis.
\begin{lemma}[Asymptotically equivalent conditions]\label{lem:asymp-equiv}
    The following conditions are equivalent in the asymptotic sense.
    \begin{enumerate}[label=(\arabic*)]
        \item \(m\ge C_1\, ps(p)\);
        \item \(m \ge C_2\, (p+1)s(p)\);
        \item \(m \ge C_3\, ps(p+1)\).
    \end{enumerate}
    In other words, if one of these conditions holds for some absolute constant \(C_i>0\), then the other two also hold for some (different) absolute constants \(C_j>0\).
\end{lemma}
\begin{proof}
    We prove (1)\(\Rightarrow\)(2)\(\Rightarrow\)(3)\(\Rightarrow\)(1) with absolute constants.

    (1)\(\Rightarrow\)(2): For \(p\ge1\), \((p{+}1)s(p) \le 2ps(p)\). Hence from \(m\ge C_1\,ps(p)\) we get
    \begin{equation}
        m\ \ge\ \tfrac{C_1}{2}\,(p{+}1)s(p),
    \end{equation}
    i.e., (2) holds with \(C_2= C_1/2\).

    (2)\(\Rightarrow\)(3): Since \(s(p)\ge s(p{+}1)\), we have \((p{+}1)s(p)\ge ps(p{+}1)\). Thus from \(m\ge C_2\,(p{+}1)s(p)\) we obtain \(m\ge C_2\,ps(p{+}1)\), i.e., (3) with \(C_3=C_2\).

    (3)\(\Rightarrow\)(1): Using \(ps(p)\le (p{+}1)s(p{+}1)\le 2\,ps(p{+}1)\), from \(m\ge C_3\,ps(p{+}1)\) we get
    \begin{equation}
        m\ \ge\ \tfrac{C_3}{2}\,ps(p),
    \end{equation}
    i.e., (1) with \(C_1=C_3/2\).

    Combining the three implications yields the claimed constant-factor equivalence among (1)-(3).
\end{proof}

We next present the classic Davis-Kahan \(\sin\Theta\) theorem~\cite{Davis1970DKsinTheta} in a convenient form.
\begin{lemma}[Davis-Kahan \(\sin\Theta\)]\label{lem:DK-sinTheta}
    Let \(\bfA\in\mathbb{R}^{p\times p}\) be symmetric with eigenvalues
    \(\lambda_1(\bfA)\ge \lambda_2(\bfA)\ge\cdots\), and let \(\bfu\) be a unit eigenvector
    for \(\lambda_1(\bfA)\). Let \(\bfE=\bfE^\top\) be a symmetric perturbation, set
    \(\hat \bfA:=\bfA+\bfE\), and let \(\hat \bfu\) be a unit top eigenvector of
    \(\hat \bfA\). If the eigengap \(\Delta:=\lambda_1(\bfA)-\lambda_2(\bfA)>0\), then
    \begin{equation}
        \sin\angle(\hat \bfu,\bfu)\ \le\ \min\Big\{1,\ \frac{\|\bfE\|_2}{\Delta}\Big\}.
    \end{equation}
    Equivalently, \(|\langle \hat \bfu,\bfu\rangle| \ \ge\ \sqrt{1-\min\{1,\ \|\bfE\|_2^2/\Delta^2\}}\).
\end{lemma}

In the following, we present two lemmas that are key to analyzing the alignment and the support stability in each round of \Cref{alg:SEP}.
\Cref{lem:DK} converts support energy on a set \(S\) into alignment with the spike. \Cref{lem:reselect} goes in the reverse direction, showing that alignment of the current round \(p\) forces the next reselected support \(S^{(p+1)}\) to capture sufficient spike energy.
These two estimates close the loop and yield a bootstrap: starting from the base energy guaranteed after the seeding step (\Cref{prop:init-gamma}), the support energy increases and the alignment improves round by round (\Cref{prop:energy-preservation}), until a good final estimator is obtained (\Cref{thm:main}).

The first lemma quantifies the alignment between the spike \(\bfv\) and the top eigenvector of any principal submatrix of \(\hGamma\), which is lower bounded in terms of the spike energy on the given support \(S\) and the noise matrix \(\bfW\).
\begin{lemma}[Alignment on a support via Davis-Kahan] \label{lem:DK}
    For all \(S\subset[n]\), \(|S|=p\), let \(\hat\bfe_S\) be the top eigenvector of \(\hGamma_{S,S}\). Then
    \begin{align}
        (x)_+ & := \max\{x,0\}, \nonumber \\
        |\langle \bfv, \hat\bfe_S\rangle|
              & \ge \|\bfv_S\|_2
        \sqrt{\Big(1-\frac{\|\bfW_{S,S}\|_{2}^2}
            {\theta^2\|\bfv_S\|_2^4}\Big)_+}.\label{eq:corr-v-eigen}
    \end{align}
\end{lemma}

\begin{proof}
    The matrix \(\hGamma_{S,S} = \theta\, \bfv_S\bfv_S^\top + \bfW_{S,S}\) has a rank-one signal part. The top eigenvector of the signal part is \(\bfu_S=\frac{\bfv_S}{\|\bfv_S\|_2}\) with eigenvalue \(\lambda_S = \theta\|\bfv_S\|_2^2\). Applying~\Cref{lem:DK-sinTheta} for the gap \(\lambda_S\) directly gives
    \begin{equation}
        \sqrt{1-|\langle \hat\bfe_S, \bfu_S\rangle|^2} = \sin\angle(\hat\bfe_S, \bfu_S) \le \frac{\|\bfW_{S,S}\|_{2}}{\theta\,\|\bfv_S\|_2^2}.
    \end{equation}
    Then, we obtain the desired bound:
    \begin{align}
        |\langle \bfv, \hat\bfe_S\rangle|
         & = |\langle \bfv_S, \hat\bfe_S\rangle| \nonumber               \\
         & = \|\bfv_S\|_2\,|\langle \bfu_S, \hat\bfe_S\rangle| \nonumber \\
         & \ge \|\bfv_S\|_2
        \sqrt{\Big(1-\frac{\|\bfW_{S,S}\|_{2}^2}
            {\theta^2\,\|\bfv_S\|_2^4}\Big)_+}.
    \end{align}
\end{proof}

In words, \Cref{lem:DK} says that more spike energy on \(S\) yields better alignment. The next lemma shows the complementary direction: better alignment forces the next top-\((p{+}1)\) support to retain sufficient spike energy.
Together they provide a closed recursion across rounds.
\begin{lemma}[Support stability of reselection based on \(\ell_2\)-norm] \label{lem:reselect}
    Under the high-probability event \(\mcE\), we have
    \begin{align}
               & \|\bfv_{S^{(p+1)}}\|_2
        \nonumber                                             \\
        \geq ~ & \sqrt{\frac{1}{s(p+1)}} -\frac{2C(1+\theta)}
        {\theta\,|\langle \bfv,\hat\bfe^{(p)}\rangle|} \sqrt{\frac{2(p+1)\log n}{m}}.\label{eq:vsplus1-lower}
    \end{align}
\end{lemma}
\begin{proof}
    Denote \(\bfv^\dagger=\theta \langle \bfv, \hat\bfe^{(p)}\rangle\bfv \) and \(\bfw^\dagger=\bfW\hat\bfe^{(p)}\). Then \(\hGamma \hat\bfe^{(p)} = \theta \langle \bfv, \hat\bfe^{(p)}\rangle \bfv + \bfW \hat\bfe^{(p)} = \bfv^\dagger + \bfw^\dagger\). Let \(T^{(p+1)}\) be the indices of the top-\((p{+}1)\) coordinates of \(|\bfv|\). By definition of \(S^{(p+1)}\) as the top-\((p{+}1)\) of \(|\hGamma \hat\bfe^{(p)}|\), we have
    \begin{equation}\label{eq:top-energy-comparison}
        \sum_{j\in S^{(p+1)}} \left|\left(\hGamma \hat\bfe^{(p)}\right)_{j}\right|^2 \ \ge\ \sum_{j\in T^{(p+1)}} \left|\left(\hGamma \hat\bfe^{(p)}\right)_{j}\right|^2.
    \end{equation}
    Since the support of \(\hat\bfe^{(p)}\) is \(S^{(p)}\), inequality~\eqref{eq:top-energy-comparison} is equivalent to
    \begin{equation}\label{eq:comparison-l2norm}
        \left\|\hGamma_{S^{(p+1)},S^{(p)}}\hat\bfe^{(p)}\right\|_2 \ \ge\ \left\|\hGamma_{T^{(p+1)},S^{(p)}}\hat\bfe^{(p)}\right\|_2.
    \end{equation}
    Set \(U := S^{(p)}\cup S^{(p+1)}\). For the LHS, we work on the support \(S^{{(p+1)}}\), and have
    \begin{align}
        \left\|\hGamma_{S^{(p+1)},S^{(p)}}\hat\bfe^{(p)}\right\|_2
         & = \|\bfv_{S^{(p+1)}}^\dagger + \bfw_{S^{(p+1)}}^\dagger\|_2 \nonumber                             \\
         & \le \|\bfv_{S^{(p+1)}}^\dagger\|_2 + \|\bfw_{S^{(p+1)}}^\dagger\|_2 \nonumber                     \\
         & \le \|\bfv_{S^{(p+1)}}^\dagger\|_2 + \|\bfW_{S^{(p+1)},S^{(p)}}\|_2\|\hat\bfe^{(p)}\|_2 \nonumber \\
         & \le \|\bfv_{S^{(p+1)}}^\dagger\|_2 + \|\bfW_{U,U}\|_2.\label{eq:LHS-bound}
    \end{align}
    The final inequality is because \(\|\bfW_{S^{(p+1)},S^{(p)}}\|_2 \le \|\bfW_{U,U}\|_2\) by zero-padding into the principal submatrix.

    Similarly, working on the support \(T^{(p+1)}\), and set \(V = S^{(p)}\cup T^{(p+1)}\), we have
    \begin{equation}\label{eq:RHS-bound}
        \left\|\hGamma_{T^{(p+1)},S^{(p)}}\hat\bfe^{(p)}\right\|_2
        \ge \|\bfv_{T^{(p+1)}}^\dagger\|_2 - \|\bfW_{V,V}\|_2.
    \end{equation}

    Combining \eqref{eq:comparison-l2norm}, \eqref{eq:LHS-bound} and \eqref{eq:RHS-bound}, we get
    \begin{equation}
        \|\bfv_{S^{(p+1)}}^\dagger\|_2 +  \|\bfW_{U,U}\|_2 \ge\|\bfv_{T^{(p+1)}}^\dagger\|_2 -  \|\bfW_{V,V}\|_2.
    \end{equation}
    Dividing by \(\theta\,|\langle \bfv,\hat\bfe^{(p)}\rangle|\) and using the high-probability event \(\mcE\), we get
    \begin{align}
               & \|\bfv_{S^{(p+1)}}\|_2
        \\
        \geq ~ & \|\bfv_{T^{(p+1)}}\|_2 -\frac{2C(1+\theta)}
        {\theta\,|\langle \bfv,\hat\bfe^{(p)}\rangle|} \sqrt{\frac{2(p+1)\log n}{m}}.
    \end{align}
    This gives the stated result~\eqref{eq:vsplus1-lower}.
\end{proof}
Combining \Cref{lem:DK} with \Cref{lem:reselect}, and instantiating the base case from \Cref{prop:init-gamma}, we obtain a monotone improvement of \(\|\mathbf v_{S^{(p)}}\|_2\) up to the entrance threshold; once the threshold is met, \Cref{prop:energy-preservation} gives the energy lower bound preservation and thus underpins the final error guarantee in \Cref{thm:main}.

\section{Proofs of Propositions}\label{app:proof-props}
\subsection{Proof of \Cref{prop:principal-submatrix}}
\begin{proof}
    Fix \(p\in[n]\) and \(S\subset[n]\), \(|S|=p\). Write
    \begin{align}
        \bfy_i & := \bSigma_{S,S}^{-1/2}\bfx_i(S),         \\
        \bfW_{S,S}
               & = \hSigma_{S,S} - \bSigma_{S,S} \nonumber \\
               & = \bSigma_{S,S}^{1/2}
        \Big(\frac{1}{m}\sum_{i=1}^m \bfy_i \bfy_i^\top-I_p\Big)
        \bSigma_{S,S}^{1/2}.
    \end{align}
    The \(\bfy_i\) are \iid isotropic sub-Gaussian in \(\R^p\), and \(\|\bSigma_{S,S}\|_{2}\le \|\bSigma\|_{2}\le 1+\theta\). Standard sub-Gaussian covariance deviation (see \cite[Theorem~4.6.1, Eq.~(4.22)]{Vershynin2018High}, with the change of variables \(t^2 \to t\)) yields for any \(t>0\),
    \begin{equation}
        \Prob\left( \big\|\tfrac{1}{m}\sum_{i=1}^m \bfy_i \bfy_i^\top - \bfI_p\big\|_{2} > C\Big(\sqrt{\tfrac{p}{m}} + \sqrt{\tfrac{t}{m}}\Big) \right) \le 2e^{-ct}.
    \end{equation}
    Taking a union bound over all \(S\) with \(|S|=p\) (there are \(\binom{n}{p}\le (en/p)^p\) choices) and over \(p\in[n]\), and choosing \(t >  \log n + p\log(en/p)\), we obtain with probability at least \(1-n^{-c}\), simultaneously for all such \(S\) and \(p\),
    \begin{align}
        \|\bfW_{S,S}\|_{2}
         & \le C(1+\theta)\Big(
        \sqrt{\tfrac{p}{m}}
        + \sqrt{\tfrac{\log n+p\log(en/p)}{m}}\Big) \nonumber \\
         & \le C(1+\theta)\sqrt{\tfrac{p\log n}{m}}.
    \end{align}
\end{proof}

\subsection{Proof of \Cref{prop:power-law-interp}}
\begin{proof}
    Let \(H_{k,\alpha} := \sum_{i=1}^k i^{-\alpha}\) (generalized harmonic number) and \(H_{p,\alpha}:=\sum_{i=1}^p i^{-\alpha}\).
    Under the normalization \(\sum_{i=1}^k v_{(i)}^2=1\) we have
    \begin{equation}
        v_{(i)}^2=\frac{i^{-\alpha}}{H_{k,\alpha}},\qquad
        s(p)=\frac{H_{k,\alpha}}{H_{p,\alpha}},
    \end{equation}
    and hence
    \begin{equation}
        p\,s^2(p)\;=\;p\Big(\frac{H_{k,\alpha}}{H_{p,\alpha}}\Big)^{\!2}.
    \end{equation}

    We use the integral test to give bounds on \(H_{k,\alpha}\).
    For \(f(x)=x^{-\alpha}\), \(f\) is positive and decreasing on \([1,\infty)\). For every integer \(i\ge 1\),
    \begin{equation}
        \int_{i}^{i+1} f(x)\,dx \;\le\; f(i) \;\le\; \int_{i-1}^{i} f(x)\,dx.
    \end{equation}
    Summing over \(i=1,\ldots,k\) gives
    \begin{equation}\label{eq:sandwich}
        \int_{1}^{k+1} x^{-\alpha}dx \;\le\; H_{k,\alpha} \;\le\; 1+\int_{1}^{k} x^{-\alpha}dx.
    \end{equation}
    Evaluating the integrals yields:
    \begin{equation}
        H_{k,\alpha}
        \asymp
        \begin{cases}
            \dfrac{k^{1-\alpha}}{1-\alpha}, & 0<\alpha<1, \\[6pt]
            1 + \log k,                     & \alpha=1,   \\[4pt]
            1,                              & \alpha>1,
        \end{cases}
    \end{equation}
    and thus
    \begin{equation}
        s(p)=\frac{H_{k,\alpha}}{H_{p,\alpha}}
        \asymp
        \begin{cases}
            \Big(\dfrac{k}{p}\Big)^{1-\alpha}, & 0<\alpha<1, \\[8pt]
            \dfrac{1+\log k}{1+\log p},        & \alpha=1,   \\[6pt]
            1,                                 & \alpha>1.
        \end{cases}\label{eq:power-law-sp}
    \end{equation}

    \begin{itemize}[itemsep=2pt, topsep=2pt]
        \item \textit{Case I: \(0\le \alpha<1\).}
              From the bounds above,
              \begin{equation}
                  p\,s^2(p)\;\asymp\; k^{2(1-\alpha)}\, p^{\,2\alpha-1}.
              \end{equation}
              Let \(g(p):=p^{2\alpha-1}\). If \(\alpha<\tfrac12\), then \(g\) decreases in \(p\), so \(\max_p p s^2(p)\) is attained at \(p=1\) and equals \(\asymp k^{2-2\alpha}\). If \(\alpha>\tfrac12\), then \(g\) increases in \(p\), so the maximum is at \(p=k\) and equals \(\asymp k^{2(1-\alpha)}k^{2\alpha-1}=k\). At \(\alpha=\tfrac12\), \(g(p)\equiv 1\), hence \(\max_p p s^2(p)\asymp k\).

        \item \textit{Case II: \(\alpha=1\).}
              We have
              \begin{equation}
                  p\,s^2(p)\;\asymp\; p\ \left(\frac{1+\log k}{1+\log p}\right)^{\!2}.
              \end{equation}
              It is easy to see that LHS attains its maximum at \(p=k\). Therefore \(\max_p p s^2(p)\asymp k\).

        \item \textit{Case III: \(\alpha>1\).} It is trivial that
              \begin{equation}
                  \max_p\ p\,s^2(p)\; \asymp \;\max_p\ p = k.
              \end{equation}

    \end{itemize}
    Combining the three cases proves the proposition.
\end{proof}

\subsection{Proof of \Cref{prop:init-gamma}}
\begin{proof}
    Working on the high-probability event \(\mcE\) with \(p=1\), we have uniformly over \(j\in[n]\) that
    \begin{equation}\label{eq:single-coordinate}
        \max_{j}\,|W_{jj}|\;=\;\max_{|S|=1}\,\|\bfW_{S,S}\|_2\;\le\;C(1+\theta)\sqrt{\tfrac{\log n}{m}}.
    \end{equation}

    Let \(l\in\arg\max_j |v_j|\) so that \(v_l^2=v_{(1)}^2\). Recall \(d_j=\hGamma_{jj}=\theta v_j^2+W_{jj}\) and \(S^{(1)}=\{\arg\max_j |d_j|\}\). Then
    \begin{equation}
        |\theta v_{S^{(1)}}^2+W_{S^{(1)},S^{(1)}}|\;\ge\;|\theta v_l^2+W_{ll}|\;\ge\;\theta v_l^2-|W_{ll}|,
    \end{equation}
    while also \(|\theta v_{S^{(1)}}^2+W_{S^{(1)},S^{(1)}}|\le\theta v_{S^{(1)}}^2+|W_{S^{(1)},S^{(1)}}|\). Hence
    \begin{equation}
        \theta v_{S^{(1)}}^2\;\ge\;\theta v_l^2-|W_{ll}|-|W_{S^{(1)},S^{(1)}}|\;\ge\;\theta v_{(1)}^2-2\max_j |W_{jj}|,
    \end{equation}
    so using~\eqref{eq:single-coordinate} we have,
    \begin{equation}
        v_{S^{(1)}}^2\;\ge\;v_{(1)}^2-\frac{2C(1+\theta)}{\theta}\sqrt{\frac{\log n}{m}}.
    \end{equation}
    Choosing
    \begin{equation}
        m\ge C\,\frac{(1+\theta)^2\log n}{\theta^2 v_{(1)}^4(1-\gamma)^2}=C\,\frac{(1+\theta)^2}{\theta^2(1-\gamma)^2}s^2(1)\log n,
    \end{equation} yields \(v_{S^{(1)}}^2\ge\gamma v_{(1)}^2\), i.e.,
    \begin{equation}
        \|\bfv_{S^{(1)}}\|_2\ge\sqrt{\gamma/s(1)}.
    \end{equation}
\end{proof}

\subsection{Proof of \Cref{prop:energy-preservation}}\label{app:proof-one-step-floor}
\begin{proof}
    For any \(p=1,\dots, k-1\), consider the current support estimate \(S^{(p)}\) and the corresponding unit vector \(\hat\bfe^{(p)}\). By the event \(\mcE\) and \Cref{lem:DK}, we have
    \begin{align}
        \nonumber     |\langle \bfv, \hat\bfe^{(p)}\rangle| & \ge \|\bfv_{S^{(p)}}\|_2\sqrt{1-\frac{(C^2(1+\theta)^2 p\log n)/m}{\theta^2\,\|\bfv_{S^{(p)}}\|_2^4}} \\
        \nonumber                                           & \ge \sqrt{\frac{\gamma}{s(p)}}\sqrt{1-\frac{(C^2(1+\theta)^2 p\log n)/m}{\theta^2 \gamma^2/s^2(p)}}   \\
                                                            & = \sqrt{\frac{\gamma}{s(p)}}\sqrt{1-\frac{C^2 (1+\theta)^2p s^2(p) \log n}{\theta^2 \gamma^2 m}}.
    \end{align}
    Choose
    \begin{equation}\label{eq:m-bound-inner}
        m \ge \frac{2C^2(1+\theta)^2 p s^2(p) \log n}{\theta^2 \gamma^2},
    \end{equation}
    then we can ensure \(|\langle \bfv, \hat\bfe^{(p)}\rangle| \ge \sqrt{\frac{\gamma}{2s(p)}}\).

    From \Cref{lem:reselect},
    \begin{align}
              & \|\bfv_{S^{(p+1)}}\|_2
        \nonumber                                                                                          \\
        \ge ~ & \sqrt{\frac{1}{s(p+1)}}
        - \frac{2\sqrt{2(p+1)}}
        {\theta\,|\langle \bfv,\hat\bfe^{(p)}\rangle|} \cdot C(1+\theta)\sqrt{\frac{\log n}{m}} \nonumber  \\
        \ge ~ & \sqrt{\frac{1}{s(p+1)}}
        - \frac{4\sqrt{(p+1)s(p)}}{\theta\sqrt{\gamma}} \cdot C(1+\theta)\sqrt{\frac{\log n}{m}} \nonumber \\
        \ge ~ & \sqrt{\frac{1}{s(p+1)}}
        - \frac{4\,C(1+\theta)}{\theta\sqrt{\gamma}} \sqrt{\frac{(p+1)s(p)\log n}{m}}.
    \end{align}
    Choose
    \begin{equation}\label{eq:m-bound-reselect}
        m \ge \frac{16C^2(1+\theta)^2 (p+1) s(p) s(p+1) \log n}{\theta^2 \gamma(1-\sqrt{\gamma})^2},
    \end{equation}
    then we can ensure that
    \begin{align}
        \|\bfv_{S^{(p+1)}}\|_2
         & \ge \sqrt{\frac{1}{s(p+1)}}
        - \sqrt{\frac{1}{s(p+1)}}(1-\sqrt{\gamma}) \nonumber         \\
         & \ge \sqrt{\frac{\gamma}{s(p+1)}}.\label{eq:vsplus1-final}
    \end{align}
    This is the desired result. Now we collect the sample size requirements from \eqref{eq:m-bound-inner} and \eqref{eq:m-bound-reselect}. Since \(0\le\gamma\le 1\), \(s(p)\ge s(p+1)\) and \(ps(p) \le (p+1)s(p+1)\), we need
    \begin{equation}
        m \ge \frac{C(1+\theta)^2 (p+1) s(p) s(p+1) \log n}{\theta^2 \gamma^2(1-\sqrt{\gamma})^2}.
    \end{equation}
    Employing \Cref{lem:asymp-equiv}, we can replace this by a cleaner uniform condition (only changing the constant \(C\)):
    \begin{equation}
        m \ge \frac{C'(1+\theta)^2 }{\theta^2 \gamma^2(1-\sqrt{\gamma})^2}(p+1) s^2(p+1)\log n.
    \end{equation}

\end{proof}

\subsection{Proof of~\Cref{prop:alpha0-lb}}
\begin{proof}
    Let \(S=S^{(k)}\). By the energy lower bound, \(\|\bfv_S\|_2\ge \sqrt{\gamma}\).
    Note that \(\bfw^{(0)}\) is the top eigenvector of \(\hGamma_{S,S}\). Applying~\Cref{lem:DK} yields
    \begin{equation}
        \alpha_0
        = |\langle \bfv, \bfw^{(0)}\rangle|
        \ \overset{\eqref{eq:corr-v-eigen}}{\ge}\ \|\bfv_S\|_2\sqrt{\Big(1-\frac{\|\bfW_{S,S}\|_2^2}{\theta^2\,\|\bfv_S\|_2^4}\Big)_{+}}.
    \end{equation}
    From the event \(\mcE\), we have \(\|\bfW_{S,S}\|_2\le C_0(1+\theta)\sqrt{k\log n/m}\). Since \(\|\bfv_S\|_2\ge\sqrt{\gamma}\), it holds that
    \begin{equation}
        \alpha_0 \ \ge\ \sqrt{\gamma}\,\sqrt{\,1-\frac{C_1(1+\theta)^2}{\theta^2\gamma^2}\cdot\frac{k\log n}{m}\,}.
    \end{equation}
    If \(m\ge C\frac{(1+\theta)^2}{\theta^2\gamma^2}k\log n\) with \(C\) sufficiently large, we have \(\alpha_0\ge \frac{1}{2}\sqrt{\gamma} \geq c_0\gamma\) with \(c_0\le 1/2 \) since \(0<\gamma<1\).
\end{proof}

\subsection{Proof of~\Cref{prop:HT-align}}

\begin{proof}
    By the sign convention in the statement of~\Cref{prop:HT-align}, we may assume without loss of generality that \(\langle\bfw,\bfv\rangle=\alpha\).
    Recall that \(\bfy=\hGamma\bfw=\theta\alpha\bfv+\bfxi\) and \(\bfxi = \bfW\bfw\). Let \(S^\star=\supp(\bfv)\) (so \(|S^\star|\le k\le k'\)), \(S^\dagger = S^\star \cup \supp(\bfw)\), and define
    \(K=\operatorname{Top}_{k'}(|\bfy|)\) using the fixed deterministic tie rule. With projection \(\mcP_K\), we then have \(\mcH_{k'}(\bfy)=\mcP_K(\bfy)\).
    Set \(K^\dagger = K \cup \supp(\bfw)\). Then \(|S^\dagger|\le 2k'\) and \(|K^\dagger|\le 2k'\) since \(|\supp(\bfw)|\le k'\).

    First, we bound the cosine angle. By definition,
    \begin{equation}
        \cos\angle\left(\frac{\mcH_{k'}(\bfy)}{\|\mcH_{k'}(\bfy)\|_2},\bfv\right)
        =\frac{|\langle \mcH_{k'}(\bfy),\bfv\rangle|}{\|\mcH_{k'}(\bfy)\|_2}.
    \end{equation}
    We next show that the inner product in the numerator is positive, so its absolute value can be dropped.
    \emph{Step 1: Lower bound for the numerator.}
    Because \(\bfv\) is supported on \(S^\star\),
    \begin{align}
        \langle \mcH_{k'}(\bfy),\bfv\rangle
         & = \langle \bfy,\bfv\rangle
        - \langle \bfy_{(K)^{\!\perp}},\bfv\rangle \nonumber \\
         & = \theta\alpha + \langle \bfxi,\bfv\rangle
        - \langle \bfy_{S^\star\setminus K},\bfv\rangle .\label{eq:decomp-Hy-v}
    \end{align}

    By the event \(\mcE\) on \(S^\dagger\),
    \begin{align}
        |\langle \bfxi,\bfv\rangle|
         & =|\langle \bfxi_{S^\star},\bfv_{S^\star}\rangle| \nonumber \\
         & \le \|\bfW_{S^\dagger, S^\dagger}\|_2
        \|\bfw\|_2 \nonumber                                          \\
         & \le C(1+\theta)\sqrt{\tfrac{k'\log n}{m}}
        = b.\label{eq:bound-xiv}
    \end{align}
    For \(\|\bfy_{S^\star\setminus K}\|_2\), the standard one-to-one pairing from \(S^\star\setminus K\) to \(K\setminus S^\star\) ensures
    \(\|\bfy_{S^\star\setminus K}\|_2 \le \|\bfy_{K\setminus S^\star}\|_2\).
    Since \(\bfy_j=\bfxi_j\) for \(j\in K\setminus S^\star\) and \(|K\setminus S^\star|\le k'\),
    another application of the high-probability event \(\mcE\) on \(T^\dagger=(K\setminus S^\star)\cup\supp(\bfw)\) (with \(|T^\dagger|\le 2k'\)) gives
    \begin{align}
        \left|\langle \bfy_{S^\star\setminus K},\bfv\rangle\right|
         & \leq \|\bfy_{S^\star\setminus K}\|_2\|\bfv\|_2 \nonumber          \\
         & \le \|\bfxi_{K\setminus S^\star}\|_2 \le b.\label{eq:bound-ysk-v}
    \end{align}

    Therefore, substituting \eqref{eq:bound-xiv} and \eqref{eq:bound-ysk-v} into \eqref{eq:decomp-Hy-v} yields
    \begin{equation}\label{eq:num-lb}
        \langle \mcH_{k'}(\bfy),\bfv\rangle \ \ge\ \theta\alpha - 2b.
    \end{equation}

    \emph{Step 2: Upper bound for the denominator.}
    We have
    \begin{equation}
        \|\mcP_K(\theta\alpha\bfv)\|_2 \le \theta\alpha,
    \end{equation}
    and
    \begin{align}
        \|\mcP_K(\bfxi)\|_2
         & = \|\bfxi_{K}\|_2 \nonumber                   \\
         & \le \|\bfW_{K^\dagger, K^\dagger}\|_2
        \|\bfw\|_2 \nonumber                             \\
         & \le C(1+\theta)\sqrt{\frac{k'\log n}{m}} = b.
    \end{align}

    Then, by triangle inequality and the same submatrix bound,
    \begin{align}
        \|\mcH_{k'}(\bfy)\|_2
         & \le \|\mcP_K(\theta\alpha\bfv)\|_2
        + \|\mcP_K(\bfxi)\|_2 \nonumber              \\
         & \le \theta\alpha + b.\label{eq:dom-bound}
    \end{align}
    Dividing the two bounds~\eqref{eq:num-lb} and \eqref{eq:dom-bound} gives the desired result~\eqref{eq:cos-lb}
    \begin{equation}
        \cos\angle\left(\frac{\mcH_{k'}(\bfy)}{\|\mcH_{k'}(\bfy)\|_2},\,\bfv\right)
        \ \ge\ \frac{\theta\alpha-2b}{\theta\alpha+b}.
    \end{equation}

    Next, we bound the sine angle. We will reuse some results from the previous part.

    \emph{Step 1: Control of the orthogonal component.}
    Let \(\bfr:=(\mathbf I-\bfv\bfv^\top)\mcH_{k'}(\bfy)\), so that
    \begin{equation}
        \sin\angle\left(\frac{\mcH_{k'}(\bfy)}{\|\mcH_{k'}(\bfy)\|_2},\bfv\right)=\frac{\|\bfr\|_2}{\|\mcH_{k'}(\bfy)\|_2}.
    \end{equation}
    Decompose
    \begin{equation}
        \bfr
        =\bigl(\mcH_{k'}(\bfy)-\theta\alpha\bfv\bigr) \;+\; \bigl(\theta\alpha-\langle\bfv,\mcH_{k'}(\bfy)\rangle\bigr)\bfv .
    \end{equation}
    By the triangle inequality and \eqref{eq:num-lb},
    \begin{equation}
        \|\bfr\|_2
        \le \|\mcH_{k'}(\bfy)-\theta\alpha\bfv\|_2 + |\theta\alpha
        -\langle\bfv,\mcH_{k'}(\bfy)\rangle|.\label{eq:r-first}
    \end{equation}
    The second term is bounded by \(2b\) since
    \begin{align}
        |\theta\alpha-\langle\bfv,\mcH_{k'}(\bfy)\rangle|
         & ~~~~\overset{\mathclap{\eqref{eq:decomp-Hy-v}}}{=}~~~~
        |\langle \bfxi,\bfv\rangle
        - \langle \bfy_{S^\star\setminus K},\bfv\rangle| \nonumber                           \\
         & ~~~~\overset{\mathclap{\eqref{eq:bound-xiv},\eqref{eq:bound-ysk-v}}}{\le}~~~~ 2b.
    \end{align}
    We now bound the first term. Note that
    \begin{align}
        \|\mcH_{k'}(\bfy)-\theta\alpha\bfv\|_2
         & =\|\mcP_K(\bfy)-\theta\alpha\bfv\|_2 \nonumber \\
         & \le \|\mcP_K(\bfxi)\|_2
        + \|(\mathbf I-\mcP_K)(\theta\alpha\bfv)\|_2.
    \end{align}
    The first term satisfies \(\|\mcP_K(\bfxi)\|_2=\|\bfxi_K\|_2\le \|\bfW_{K^\dagger, K^\dagger}\|_2\|\bfw\|_2 \le b\) by the high-probability event \(\mcE\).
    For the second term, using \(\bfy=\theta\alpha\bfv+\bfxi\),
    \begin{align}
        \|(\mathbf I-\mcP_K)(\theta\alpha\bfv)\|_2
         & =\|(\theta\alpha\,\bfv)_{K^\perp}\|_2 \nonumber \\
         & =\|\bfy_{S^\star\setminus K}
        -\bfxi_{S^\star\setminus K}\|_2 \nonumber          \\
         & \le \|\bfy_{S^\star\setminus K}\|_2
        + \|\bfxi_{S^\star\setminus K}\|_2 \nonumber       \\
         & \le b+b=2b,
    \end{align}
    where both terms are bounded by \(b\) using the high-probability event \(\mcE\).
    Hence
    \begin{equation}\label{eq:A-3b}
        \|\mcH_{k'}(\bfy)-\theta\alpha\bfv\|_2\ \le\ b+2b\ =\ 3b.
    \end{equation}
    Substituting \eqref{eq:A-3b} into \eqref{eq:r-first} gives
    \begin{equation}
        \|\bfr\|_2 \ \le\ 5 b.
    \end{equation}

    \emph{Step 2: Lower bound for the denominator.}
    From \eqref{eq:num-lb},
    \(
    \|\mcH_{k'}(\bfy)\|_2\ \ge\ \langle\bfv,\mcH_{k'}(\bfy)\rangle\ \ge\ \theta\alpha-2b.
    \)
    Therefore,
    \begin{align}
        \sin\angle\left(
        \frac{\mcH_{k'}(\bfy)}{\|\mcH_{k'}(\bfy)\|_2},\bfv\right)
         & =\frac{\|\bfr\|_2}{\|\mcH_{k'}(\bfy)\|_2} \nonumber \\
         & \le \frac{5b}{\theta\alpha-2b},
    \end{align}
    which proves \eqref{eq:sin-up}.
\end{proof}

\subsection{Proof of~\Cref{prop:invariant}}
\begin{proof}
    Base case \(t=0\) is \Cref{prop:alpha0-lb}: \(\alpha_0\ge c_0\gamma\). Set \(c_\ast = c_0 \in (0,1/2]\).
    Induction step: suppose \(\alpha_t\ge c_\ast\gamma\). Apply~\Cref{prop:HT-align} to \(\bfw^{(t)}\) with \(b=C(1+\theta)\sqrt{k'\log n/m}\). Using \(m\) as stated, we may ensure \(b\le \frac{1}{6}\theta c_\ast\gamma\). Then
    \begin{align}
        \cos\angle(\bfw^{(t+1)},\bfv)
         & \overset{\eqref{eq:cos-lb}}{\ge}
        \frac{\theta\alpha_t-2b}{\theta\alpha_t+b} \nonumber                         \\
         & \overset{\hphantom{\eqref{eq:cos-lb}}}{\ge}
        \frac{\theta c_\ast\gamma-2(\theta c_\ast\gamma/6)}
        {\theta c_\ast\gamma+(\theta c_\ast\gamma/6)} \nonumber                      \\
         & \overset{\hphantom{\eqref{eq:cos-lb}}}{=} \frac{2}{3}\,\Big/\,\frac{7}{6}
        = \frac{4}{7} > \frac{1}{2}.
    \end{align}
    Hence \(\alpha_{t+1}=\cos\angle(\bfw^{(t+1)},\bfv)\ge \tfrac12 \ge c_\ast\gamma\) for \(c_\ast\in (0,1/2]\). Thus the invariant holds for all \(t\).
\end{proof}

%% file: 09Biographies.tex
\begin{IEEEbiographynophoto}{Mengchu Xu}
  (Member, IEEE) received the B.S. degree in mathematics and the Ph.D. degree in statistics from the School of Data Science, Fudan University, Shanghai, China, in 2019 and 2024, respectively. He was a Postdoctoral Researcher with the Weizmann Institute of Science, Rehovot, Israel, working in the group headed by Prof. Yonina C. Eldar. He is currently a Postdoctoral Fellow with the School of Data Science, Fudan University, China. His research interests include high-dimensional statistics, sparse optimization, and deep learning.
\end{IEEEbiographynophoto}

\begin{IEEEbiographynophoto}{Jian Wang}
  (Member, IEEE) received the B.S. and M.S. degrees from Harbin Institute of Technology, China, in 2006 and 2009, respectively, and the Ph.D. degree in electrical and computer engineering from Korea University, South Korea, in 2013. From 2013 to 2017, he was a Post-Doctoral Researcher with the Department of Statistics, Rutgers University, and the Department of Electrical and Computer Engineering, Duke University, USA; and a Research Assistant Professor with the Department of Electrical and Computer Engineering, Seoul National University, South Korea. Since 2017, he has been with the School of Data Science, Fudan University, China. His research interests include sparse and low-rank recovery, signal processing in wireless communications, and statistical learning. He was nominated for the Young Author Best Paper Award from the IEEE Signal Processing Society.
\end{IEEEbiographynophoto}

\begin{IEEEbiographynophoto}{Yonina C. Eldar}
  (Fellow, IEEE) received the B.Sc. degree in physics and the B.Sc. degree in electrical engineering from Tel-Aviv University, and the Ph.D. degree in electrical engineering and computer science from MIT. She is the Aoun Chair Professor of electrical and computer engineering at Northeastern University and the Dorothy and Patrick Gorman Professorial Chair of mathematics and computer science at the Weizmann Institute where she founded and heads the Signal Acquisition Modeling Processing and Learning Lab (SAMPL) and the Center for Biomedical Engineering. She is also a Visiting Professor with MIT, a Visiting Research Collaborator with Princeton, a Visiting Scientist with the Broad Institute, and an Adjunct Professor with Duke University and was a Visiting Professor at Stanford. She is a member of the Israel Academy of Sciences and Humanities and of the Academia Europaea, a EURASIP Fellow. She has received many awards for excellence in research and teaching, including the Israel Prize (2025), the Landau Prize (2024), the IEEE Signal Processing Society Technical Achievement Award (2013), the IEEE/AESS Fred Nathanson Memorial Radar Award (2014), and the IEEE Kiyo Tomiyasu Award (2016). She received the Michael Bruno Memorial Award from the Rothschild Foundation, the Weizmann Prize for Exact Sciences, the Wolf Foundation Krill Prize for Excellence in Scientific Research, the Henry Taub Prize for Excellence in Research (twice), the Hershel Rich Innovation Award (three times), and the Award for Women with Distinguished Contributions. She was selected as one of the 50 most influential women in Israel, and was a member of the Israel Committee for Higher Education. She is the Editor-in-Chief of \emph{Foundations and Trends in Signal Processing}, a member of several IEEE Technical Committees and Award Committees, and heads the Committee for Promoting Gender Fairness in Higher Education Institutions in Israel.
\end{IEEEbiographynophoto}

%% file: 00Main.bbl
\begin{thebibliography}{10}
\providecommand{\url}[1]{#1}
\csname url@samestyle\endcsname
\providecommand{\newblock}{\relax}
\providecommand{\bibinfo}[2]{#2}
\providecommand{\BIBentrySTDinterwordspacing}{\spaceskip=0pt\relax}
\providecommand{\BIBentryALTinterwordstretchfactor}{4}
\providecommand{\BIBentryALTinterwordspacing}{\spaceskip=\fontdimen2\font plus
\BIBentryALTinterwordstretchfactor\fontdimen3\font minus
  \fontdimen4\font\relax}
\providecommand{\BIBforeignlanguage}[2]{{%
\expandafter\ifx\csname l@#1\endcsname\relax
\typeout{** WARNING: IEEEtran.bst: No hyphenation pattern has been}%
\typeout{** loaded for the language `#1'. Using the pattern for}%
\typeout{** the default language instead.}%
\else
\language=\csname l@#1\endcsname
\fi
#2}}
\providecommand{\BIBdecl}{\relax}
\BIBdecl

\bibitem{EckartYoung1936MatrixApprox}
C.~Eckart and G.~Young, ``The approximation of one matrix by another of lower
  rank,'' \emph{Psychometrika}, vol.~1, no.~3, pp. 211--218, 1936.

\bibitem{Jolliffe2002PCA}
I.~T. Jolliffe, \emph{Principal Component Analysis}, 2nd~ed.\hskip 1em plus
  0.5em minus 0.4em\relax New York, NY: Springer, 2002.

\bibitem{TurkPentland1991Eigenfaces}
M.~Turk and A.~Pentland, ``Face recognition using eigenfaces,'' in \emph{Proc.
  IEEE Comput. Soc. Conf. Comput. Vision Pattern Recognit. (CVPR)}, 1991, pp.
  586--591.

\bibitem{Zou2006SPCA}
H.~Zou, T.~Hastie, and R.~Tibshirani, ``Sparse principal component analysis,''
  \emph{J. Comput. Graph. Statist.}, vol.~15, no.~2, pp. 265--286, 2006.

\bibitem{Guan2009SPCAProbabilistic}
Y.~Guan and J.~G. Dy, ``Sparse probabilistic principal component analysis,'' in
  \emph{Proc. Int. Conf. Artif. Intell. Stat. (AISTATS)}, ser. Proceedings of
  Machine Learning Research, vol.~5, 2009, pp. 185--192.

\bibitem{Lee2010SPCABiclustering}
M.~Lee, H.~Shen, J.~Z. Huang, and J.~S. Marron, ``Biclustering via sparse
  singular value decomposition,'' \emph{Biometrics}, vol.~66, no.~4, pp.
  1087--1095, 2010.

\bibitem{GuerraUrzola2021SPCAGuide}
R.~Guerra-Urzola, K.~Van~Deun, J.~C. Vera, and K.~Sijtsma, ``A guide for sparse
  {PCA}: model comparison and applications,'' \emph{Psychometrika}, vol.~86,
  no.~4, pp. 893--919, 2021.

\bibitem{VuLei2012MinimaxSPCA}
V.~Q. Vu and J.~Lei, ``Minimax rates of estimation for sparse {PCA} in high
  dimensions,'' in \emph{Proc. Int. Conf. Artif. Intell. Stat. (AISTATS)}, ser.
  JMLR Proceedings, vol.~22, 2012, pp. 1278--1286.

\bibitem{WangLuLiu2014MinimaxSPCA}
Z.~Wang, H.~Lu, and H.~Liu, ``Tighten after relax: Minimax-optimal sparse {PCA}
  in polynomial time,'' in \emph{Advances in Neural Information Processing
  Systems 27}, 2014.

\bibitem{Johnstone2009SPCAConsistency}
I.~M. Johnstone and A.~Y. Lu, ``On consistency and sparsity for principal
  components analysis in high dimensions,'' \emph{J. Amer. Statist. Assoc.},
  vol. 104, no. 486, pp. 682--693, 2009.

\bibitem{BerthetRigollet2013aSPCADetection}
Q.~Berthet and P.~Rigollet, ``Optimal detection of sparse principal components
  in high dimension,'' \emph{Ann. Statist.}, vol.~41, no.~4, pp. 1780--1815,
  2013.

\bibitem{Liu2021Towards}
Z.~Liu, S.~Ghosh, and J.~Scarlett, ``Towards sample-optimal compressive phase
  retrieval with sparse and generative priors,'' \emph{Adv. Neural Inf.
  Process. Syst.}, vol.~34, 2021.

\bibitem{Ma2013SPCAIT}
Z.~Ma, ``Sparse principal component analysis and iterative thresholding,''
  \emph{Ann. Statist.}, vol.~41, no.~2, pp. 772--801, 2013.

\bibitem{dAspremont2007PCASDP}
A.~d'Aspremont, L.~El~Ghaoui, M.~I. Jordan, and G.~R.~G. Lanckriet, ``A direct
  formulation for sparse {PCA} using semidefinite programming,'' \emph{SIAM
  Rev.}, vol.~49, no.~3, pp. 434--448, 2007.

\bibitem{Amini2009SPCASDP}
A.~A. Amini and M.~J. Wainwright, ``High-dimensional analysis of semidefinite
  relaxations for sparse principal components,'' \emph{Ann. Statist.}, vol.~37,
  no.~5B, pp. 2877--2921, 2009.

\bibitem{Krauthgamer2015SDPSPCAInfoLimit}
R.~Krauthgamer, B.~Nadler, and D.~Vilenchik, ``{Do semidefinite relaxations
  solve sparse PCA up to the information limit?}'' \emph{Ann. Statist.},
  vol.~43, no.~3, pp. 1300 -- 1322, 2015.

\bibitem{Feige2000HiddenClique}
U.~Feige and R.~Krauthgamer, ``Finding and certifying a large hidden clique in
  a semirandom graph,'' \emph{Random Struct. Algorithms}, vol.~16, no.~2, pp.
  195--208, 2000.

\bibitem{BerthetRigollet2013bSPCADetectionLB}
Q.~Berthet and P.~Rigollet, ``Complexity theoretic lower bounds for sparse
  principal component detection,'' in \emph{Proc. Conf. Learn. Theory (COLT)},
  2013, pp. 1046--1066.

\bibitem{Wang2016SPCACompStatTradeoff}
T.~Wang, Q.~Berthet, and R.~J. Samworth, ``{Statistical and computational
  trade-offs in estimation of sparse principal components},'' \emph{Ann.
  Statist.}, vol.~44, no.~5, pp. 1896 -- 1930, 2016.

\bibitem{Moghaddam2006SPCABounds}
B.~Moghaddam, Y.~Weiss, and S.~Avidan, ``Spectral bounds for sparse {PCA}:
  exact and greedy algorithms,'' in \emph{Adv. Neural Inf. Process. Syst.},
  2006, pp. 915--922.

\bibitem{dAspremont2008PCAOptimalSolution}
A.~d'Aspremont, F.~Bach, and L.~El~Ghaoui, ``Optimal solutions for sparse
  principal component analysis,'' \emph{J. Mach. Learn. Res.}, vol.~9, pp.
  1269--1294, 2008.

\bibitem{Birnbaum2013MinimaxSPCA}
A.~Birnbaum, I.~M. Johnstone, B.~Nadler, and D.~Paul, ``Minimax bounds for
  sparse {PCA} with noisy high-dimensional data,'' \emph{Ann. Statist.},
  vol.~41, no.~3, pp. 1055--1084, 2013.

\bibitem{Carvalho2008SPCAGeneExpression}
C.~M. Carvalho, J.~E. Lucas, Q.~Wang, J.~Chang, J.~R. Nevins, and M.~West,
  ``High-dimensional sparse factor modeling: Applications in gene expression
  genomics,'' \emph{J. Amer. Statist. Assoc.}, vol. 103, no. 484, pp.
  1438--1456, 2008.

\bibitem{Seghouane2018AdaptiveBlockSPCA}
A.-K. Seghouane and A.~Iqbal, ``The adaptive block sparse {PCA} and its
  application to multi-subject fmri data analysis using sparse mcca,''
  \emph{Signal Process.}, vol. 153, pp. 311--320, 2018.

\bibitem{Wu2021Hadamard}
F.~Wu and P.~Rebeschini, ``Hadamard wirtinger flow for sparse phase
  retrieval,'' in \emph{Proc. Int. Conf. Artif. Intell. Statist.}, 2021, pp.
  982--990.

\bibitem{Cai2022Provable}
J.-F. Cai, J.~Li, and J.~You, ``Provable sample-efficient sparse phase
  retrieval initialized by truncated power method,'' \emph{Inverse Problems},
  vol.~39, p. 075008, 2023.

\bibitem{Xu2024Exponential}
M.~Xu, Y.~Zhang, and J.~Wang, ``Exponential spectral pursuit: An effective
  initialization method for sparse phase retrieval,'' in \emph{Proc. Int. Conf.
  Mach. Learn.}, ser. Proceedings of Machine Learning Research, vol. 235.\hskip
  1em plus 0.5em minus 0.4em\relax PMLR, Jul. 2024, pp. 55\,525--55\,546.

\bibitem{Cai25PeakPCA}
J.-F. Cai, Z.~Xian, and J.~Ying, ``Fast and provable algorithms for sparse
  {PCA} with improved sample complexity,'' in \emph{Proc. Int. Conf. Mach.
  Learn. (ICML)}, vol. 267, 13--19 Jul 2025, pp. 6319--6340.

\bibitem{Paul2007SPCAeigenstructure}
D.~Paul, ``Asymptotics of sample eigenstructure for a large dimensional spiked
  covariance model,'' \emph{Statist. Sinica}, vol.~17, no.~4, pp. 1617--1642,
  2007.

\bibitem{Yuan2013TPower}
X.-T. Yuan and T.~Zhang, ``Truncated power method for sparse eigenvalue
  problems,'' \emph{J. Mach. Learn. Res.}, vol.~14, pp. 899--925, 2013.

\bibitem{Cai2013OptimalSPCA}
T.~T. Cai, Z.~Ma, and Y.~Wu, ``Sparse {PCA}: Optimal rates and adaptive
  estimation,'' \emph{Ann. Statist.}, vol.~41, no.~6, pp. 3074--3110, 2013.

\bibitem{Xu2025OptimalPR}
M.~Xu, Y.~Zhang, and J.~Wang, ``Achieving optimal sample complexity for a
  broader class of signals in sparse phase retrieval,'' \emph{{IEEE} Trans.
  Signal Process.}, vol.~74, pp. 1735--1750, 2026.

\bibitem{Ma2020LOO}
C.~Ma, K.~Wang, Y.~Chi, and Y.~Chen, ``Implicit regularization in nonconvex
  statistical estimation: Gradient descent converges linearly for phase
  retrieval, matrix completion, and blind deconvolution,'' \emph{Found. Comput.
  Math.}, vol.~20, pp. 451--632, 2020.

\bibitem{arcones1993decoupling}
M.~A. Arcones and E.~Giné, ``On decoupling, series expansions, and tail
  behavior of chaos processes,'' \emph{J. Theor. Probab.}, vol.~6, no.~1, pp.
  101--122, 1993.

\bibitem{DiCiccio2020DataSplit}
C.~J. DiCiccio, T.~J. DiCiccio, and J.~P. Romano, ``Exact tests via multiple
  data splitting,'' \emph{Stat. Probab. Lett.}, vol. 166, p. 108865, 2020.

\bibitem{Jagatap2019Sample}
G.~Jagatap and C.~Hegde, ``Sample-efficient algorithms for recovering
  structured signals from magnitude-only measurements,'' \emph{{IEEE} Trans.
  Inf. Theory}, vol.~65, no.~7, pp. 4434--4456, 2019.

\bibitem{Hardoon2011SCCA}
D.~R. Hardoon and J.~Shawe-Taylor, ``Sparse canonical correlation analysis,''
  \emph{Mach. Learn.}, vol.~83, no.~3, pp. 331--353, 2011.

\bibitem{Davis1970DKsinTheta}
C.~Davis and W.~M. Kahan, ``The rotation of eigenvectors by a perturbation.
  iii,'' \emph{SIAM J. Numer. Anal.}, vol.~7, no.~1, pp. 1--46, 1970.

\bibitem{Vershynin2018High}
R.~Vershynin, \emph{High-Dimensional Probability: An Introduction with
  Applications in Data Science}, ser. Cambridge Series in Statistical and
  Probabilistic Mathematics.\hskip 1em plus 0.5em minus 0.4em\relax Cambridge
  University Press, 2018.

\end{thebibliography}
